\documentclass[11pt,leqno]{amsart}

\usepackage{amsmath}
\usepackage{amssymb}
\usepackage{a4wide}
\usepackage{bbm}
\usepackage{graphics}
\usepackage{epsfig}
\usepackage{hyperref}
\usepackage{comment}
\usepackage{mathtools}
\usepackage{soul}
\hypersetup{%
  colorlinks = true,
  linkcolor  = blue,
  citecolor    = red
}
 \usepackage[applemac]{inputenc}
\usepackage{amsmath, amssymb, latexsym, amscd, amsthm,amsfonts,amstext}
\usepackage[mathscr]{eucal}
\usepackage{appendix}
\usepackage[active]{srcltx}
\include{srctex}

\usepackage{color}

 \usepackage{mleftright}
\usepackage{pgf,tikz}
\usetikzlibrary{patterns,arrows,intersections, decorations.markings, decorations.pathmorphing,
	backgrounds, positioning, fit, shapes.geometric}
\usepackage{caption}  
\usepackage{mathrsfs}

\tikzset{arrow data/.style 2 args={%
      decoration={%
         markings,
         mark=at position #1 with \arrow{#2}},
         postaction=decorate}
      }%

\newcounter{hypo}

\newcommand{\Z}{\mathbb{Z}}
\newcommand{\R}{\mathbb{R}}

\newcommand{\W}{{\mathcal W}}

\newcommand{\cO}{{\mathcal O}}

\newcommand{\ord}{{\mathcal O}}

\newcommand{\boundellipse}[3]
{(#1) ellipse (#2 and #3)
}

\makeatletter
 \@addtoreset{equation}{section}
 \makeatother
 
 \usepackage{hyperref} 

\newtheorem{theorem}{Theorem}[section]
\newtheorem{lemma}[theorem]{Lemma}
\newtheorem{proposition}[theorem]{Proposition}

\newtheorem{remark}[theorem]{Remark}

\theoremstyle{definition}

\numberwithin{equation}{section}
 
  {\setlength{\baselineskip}{1.5\baselineskip}

\newtheorem{thm}{Theorem}[section]

\def\centerarc[#1](#2)(#3:#4:#5);%
    {
    \draw[#1]([shift=(#3:#5)]#2) arc (#3:#4:#5);
    }


\title[A double well problem for a system of Schr\"odinger operators]{Eigenvalue splitting for a system of Schr\"odinger operators with an energy-level crossing}
 \author[M. Assal and S. Fujii\'e]{Marouane Assal and Setsuro Fujii\'e }
\address{Marouane Assal, Facultad de Matem\'aticas, Pontificia Universidad Cat\'olica de Chile, Vicuna Mackenna 4860, Chile}
\email{marouane.assal@mat.puc.cl}
\address{Setsuro Fujii\'e, Department of Mathematical Sciences, Ritsumeikan University, 111 Noji-Higashi, Kusatsu, 525-8577,  Japan}
\email{fujiie@fc.ritsumei.ac.jp}

\begin{document}



\maketitle

\begin{abstract}
We study the asymptotic distribution of the eigenvalues of a one-dimensional two-by-two semiclassical system of coupled Schr\"odinger operators in the presence of two potential wells and with an energy-level crossing. We provide Bohr-Sommerfeld quantization condition for the eigenvalues of the system on any energy-interval above the crossing and give precise  asymptotics in the semiclassical limit $h\to 0^+$. In particular, in the symmetric case, the eigenvalue splitting occurs and we prove that the splitting is of polynomial order $h^{\frac32}$ and that the main term in the asymptotics is governed by the area of the intersection of the two  classically allowed domains.
\end{abstract}



\section{Introduction}



In this paper, we study the asymptotic distribution, in the semiclassical limit $h\to 0^+$, of the eigenvalues of one dimensional $2\times 2$ systems of coupled Schr\"odinger operators of the form 
\begin{equation}\label{Sintro} P(h) := 
\begin{pmatrix}
P_1(h) & h W\\\\
h W^* & P_2(h)
\end{pmatrix} \;\;\; {\rm in }\;\; L^2(\mathbb R)\oplus L^2(\mathbb R),
\end{equation}
where $P_1(h),P_2(h)$ are semiclassical Schr\"odinger operators on the real line 
$$
P_j(h):= -h^2 \frac{d^2}{dx^2} + V_j(x) \;\;\; (j=1,2),
$$
and the interaction operator $W$ is a first order semiclassical differential operator, and $W^*$ its formal adjoint. Matrix Schr\"odinger operators arise as important models in molecular physics and quantum chemistry, for instance in the Born-Oppenheimer approximation which allows for a drastic reduction of problem size when dealing with molecular systems. In this context, systems of the form \eqref{Sintro} appear in a natural way after a Born-Oppenheimer reduction of diatomic molecular Hamiltonians, in which case, the semiclassical parameter $h$ represents the square root of the quotient between the electronic and nuclear masses (see e.g. \cite{KMSW}).

We are interested in the situation where the potentials $V_1$ and $V_2$ are smooth functions on $\mathbb R$, cross at some point, and each of them presents a simple well at the considered energy level (see Figure \ref{Fig1}). In particular, in the phase space, the characteristic set is the union of two closed curves $\Gamma_j(E)=\{(x,\xi)\in T^*\R;\xi^2+V_j(x)=E\}$, $j=1,2$, which intersect transversally at two points (see Figure \ref{Fig2}). For such a model, we study the asymptotic distribution of the eigenvalues of the operator $P(h)$ at energies above that of the crossing, focusing mainly on the asymptotic of the splitting in the case of symmetric potentials.

In the scalar case, the study of the eigenvalues of the one-dimensional Schr\"odinger operator in the presence of potential wells goes back to the beginning of quantum mechanics and many results has been rigorously proved since then. In the case of a simple potential well, it is well known that (see for instance \cite{Fe} in the case of analytic potentials and \cite{Ya, ILR} and references therein in the $C^{\infty}$ case) the eigenvalues are approximated, in the semiclassical limit $h\to 0^+$, by the roots of the so called Bohr-Sommerfeld quantization rule
$e^{2i\mathcal{A}(E)/h}+1=0$, where $2\mathcal{A}(E)$ is  the line integral $\int_{\Gamma(E)}\xi dx$ on the closed curve (energy surface) $\Gamma(E)=\{(x,\xi)\in T^*\R;\xi^2+V(x)=E\}$ 
in the phase space. When the potential has a double well (or more generally multiple wells) at an energy level, the eigenvalues near this level are approximated (with an exponentially small error) by the union of those created by each well (see \cite{HeSj1}). In particular, if the potential is even, that is the wells are symmetric, these two sets of eigenvalues coincide and a quantum tunneling effect between the two wells appears as eigenvalue splitting (see Landau-Lifschitz \cite{LF}). More precisely, in this case the eigenvalues of the operator come out in pairs with splitting exponentially small and the main term in the asymptotic as $h\to 0^+$ of the difference between these two eigenvalues is given by $he^{-S(E)/h}/\mathcal{A}'(E)$ where $S(E)$ is the Agmon distance between the two wells. This fact was rigorously proved by means of the exact WKB method by G\'erard and Grigis \cite{GG} assuming analyticity on the potential. The same problem has recently been studied by Hirota and Wittsten \cite{HW} for complex eigenvalues of a non-self-adjoint Dirac operator (Zakharov-Shabat operator).

A similar double well problem can be considered in 
the coupled Schr\"odinger case below the crossing energy. 
Pettersson \cite{Pe} studied the eigenvalue splitting  at the bottom of the potentials (with a weaker interaction), which is below a crossing. In this case, the splitting of eigenvalue is still exponentially small, given by a tunneling effect between the wells.

Above the crossing energy, where the two closed curves $\Gamma_1$ and $\Gamma_2$ intersect, the same problem can be considered, but
the quantum transition at crossing points plays the essential role instead of tunneling. 
Krivko and Kucherenko \cite{KK} proved, as an analogy to the scalar case with double well,  that the eigenvalues of the system are approximated by the union of those of each scalar operator. 
We give here
a Bohr-Sommerfeld type quantization rule with precise estimate  of the error (Theorem \ref{thQCS}). From this rule we recover the above result as well as the asymptotic behavior of the splitting (Theorem \ref{mainth}).
The splitting is of polynomial order $h^{3/2}$ and the principal term is governed by the area of the intersection of the domain bounded by these two curves (see Figure \ref{Fig2}). 

Before ending this section, we explain briefly the main ideas in our study. Our method consists of two parts. In the first part, we construct by iteration exact decaying solutions to the system $P(h) u = E u$ on $\R_+$ and  on $\R_-$, starting from solutions to the underlying scalar equations $P_j(h) u = E u$, $j=1,2$. This method of construction was established in \cite{FMW1,FMW2} for the study of resonances of a system of the form \eqref{Sintro}, where one of the two operators $P_1$ or $P_2$ is ``non-trapping" at the considered energy.
The quantization rule will then be given by a linear dependence condition on these solutions. This allows us to prove the existence of eigenvalues of $P(h)$ together with a rough estimate on their location (Theorem \ref{ThapproxRes}). In the second part, we improve the quantization rule that determines the eigenvalues of $P(h)$ in the considered energy interval and we obtain a better estimate on their location, using a purely microlocal approach. First, we start by the construction of a basis of microlocal solutions to the system $(P(h)-E)u=0$ on each component  of the characteristic set divided by the crossing and turning points (see Figure \ref{Fig 3}) using the standard WKB constructions. Then, the main step consists in deriving microlocal connection formulae that determine the propagation of microlocal data through the crossing and turning points. These connection problems were studied in \cite{FMW3} for the application to a resonance problem (see also \cite{Hi}). Here we give the transfer matrix at a crossing point in a general setting for semiclassical pseudodifferential systems (Theorem \ref{TMF}). We mention that the microlocal study of solutions at such a crossing point had also been done by Colin de Verdi\`ere \cite{Cdv1, Cdv2, Cdv3} using the Landau-Zener normal form
(see also Kucherenko \cite{Ku} for the multiphase method).

\section{Assumptions and results}

Consider the semiclassical $2\times2$ matrix Schr\"odinger operator
\begin{equation}\label{System0} P(h) := 
\begin{pmatrix}
 P_1(h) & hW\\\\
hW^* & P_2(h)
\end{pmatrix},
\end{equation}
on the Hilbert space $L^2(\mathbb R)\oplus L^2(\mathbb R)$, with 
$$
P_j(h):=-h^2 \frac{d^2}{dx^2} + V_j(x) \quad (j=1,2),
$$
where $V_1,V_2$ are real-valued potentials on the real line, $W$ is a first-order semiclassical differential operator, and $W^*$ its formal adjoint. Here $h>0$ is the semiclassical parameter, and we work in the semiclassical regime $h\rightarrow 0^+$.

We make the following assumptions on the potentials $V_1$, $V_2$ and the interaction operator $W$. Let $I_0:= (E_1,E_2)$ be an energy interval with $0<E_1<E_2\in \mathbb R$.

\vspace{0.3cm}

\noindent
{\bf Assumption (A1).} $V_1,V_2$ are smooth and real-valued
on $\R$, and satisfy the following conditions (see Figure \ref{Fig1})
\begin{itemize}
\item[i)] $V_1,V_2$ admit limits as $x\to \pm\infty$  greater than $E_2$. 
\item[ii)] For all $E\in I_0$, there exist four  numbers $\alpha_1(E)<\alpha_2(E)<0<\beta_1(E)<\beta_2(E)$ such that, for $j=1,2$, 
\begin{equation*}
\left\{\begin{array}{lll}
V_j>E  \;\; &{\rm on} \;\;  (-\infty, \alpha_j(E))\cup (\beta_j(E),+\infty) \\
V_j<E \;\; &{\rm on} \;\;(\alpha_j(E),\beta_j(E))
\end{array}\right. \quad {\rm and} \quad  V_j'(\alpha_j(E))<0,\;  V_j'(\beta_j(E))>0.
\end{equation*}
\item[iii)] The set $\{V_1=V_2 < E_2 \}$ is reduced to $\{0\}$, and one has 
$$
V_1(0)=V_2(0)=0, \;\;\; V'_1(0)>0, \; V_2'(0)<0.
$$
\end{itemize}

Thus $(\alpha_1(E),\beta_1(E))$ and $(\alpha_2(E),\beta_2(E))$ are two potential wells associated with $V_1$ and $V_2$ respectively, and energy-level crossing occurs at one point, the origin, below $E$.

\begin{center}
\begin{figure}[h]
	\begin{tikzpicture}[scale=0.5][hick,>=stealth',dot/.style = {
		draw,
		fill = white,
		circle,
		inner sep = 0pt,
		minimum size = 2pt}]
	\coordinate (O) at (0,0);
	\draw[very thick,->] (-9,0) -- (11,0) coordinate[label = {below:\tiny{$x$}}] (xmax);
	\draw[very thick,->] (0.5,-4) -- (0.5,5) coordinate[label = {right:\tiny{$V(x)$}}] (ymax);
	\draw[very thick] plot[smooth] coordinates {(-9,3.5)(-5.4,2.3) (-3.34,-1.32) (-1.78,-2.47)(1,0.5)(3.76,2.51)(10.9,3)};
	\draw[very thick] plot[smooth] coordinates {(-9,2.3)(-8,2.3)(-4.6,3.42)(4.4,-2.04)(9.3,2.8)(10.9,3.5)};
	\draw[thick,dashed] (-4.65,0) -- (-4.65,1.29);
	\node at (-4.9,-0.4) {\tiny {$\alpha_1(E)$}};
	\draw[thick,dashed] (-1.25,0) -- (-1.25,1.29);
	\node at (-1.45,-0.4) {\tiny {$\alpha_2(E)$}};
	\draw[thick,dashed] (1.73,0) -- (1.73,1.29);
	\node at (2.5,-0.4) {\tiny {$\beta_1(E)$}};
	\draw[thick,dashed] (8.1,0) -- (8.1,1.29);
	\node at (8.25,-0.4) {\tiny {$\beta_2(E)$}};
	\node at (4,3) {\tiny {$V_1$}};
	\node at (-4.7,4) {\tiny {$V_2$}};
	\node at (11.2,1.2) {\tiny {$E$}};

	\draw[very thick] (-9,1.2) -- (11,1.2);

	\node at (0.69,-0.4) {\tiny {$0$}};
		\end{tikzpicture}
		\caption{The potentials $V_1,V_2$} \label{Fig1}
	\end{figure}
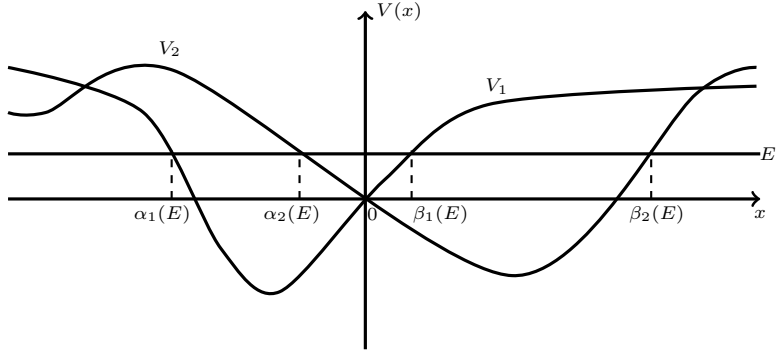
\end{center}

Let $p_j(x,\xi):= \xi^2+ V_j(x)$ be the semiclassical symbol of $P_j(h)$, $j=1,2$. In the phase space $\mathbb R^2=\mathbb R_x\times \mathbb R_{\xi}$, the characteristic sets 
$$
\Gamma_j(E):=\big\{(x,\xi)\in \mathbb R^{2}; \, p_j(x,\xi)=E\big\} \quad (j=1,2),
$$
are closed curves which intersect transversally at two (crossing) points $\rho_{\pm}(E):=(0,\pm\sqrt{E})$ (see Figure \ref{Fig2}).

\begin{center}
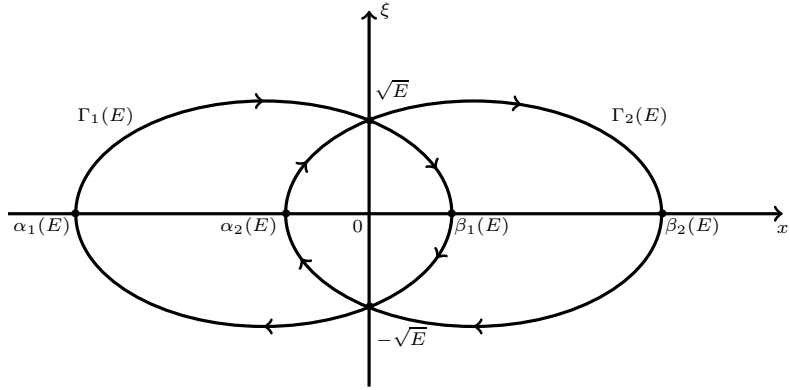
\begin{figure}[h]
	\begin{tikzpicture}[scale=1]
	\coordinate (O) at (0,0);
	\draw[very thick,->] (-4.3,0) -- (6,0) coordinate[label = {below:\tiny{$x$}}] (xmax);
	\draw[very thick,->] (0.5,-2.3) -- (0.5,2.7) coordinate[label = {right:\tiny{$\xi$}}] (ymax);
	\begin{scope}[decoration={markings,
	mark=at position 0.05 with {\arrowreversed[very thick]{>}},
    mark=at position 0.25 with {\arrowreversed[very thick]{>}},
     mark=at position 0.75 with {\arrowreversed[very thick]{>}},
                   mark=at position 0.95 with {\arrowreversed[very thick]{>}},
    },
            ]
	\draw[very thick,postaction={decorate}] \boundellipse{-0.9,0}{2.5}{1.5} ;
	\end{scope}
	\begin{scope}[decoration={markings,
    mark=at position 0.44 with {\arrowreversed[very thick]{>}},
        mark=at position 0.20 with {\arrowreversed[very thick]{>}},
    mark=at position 0.55 with {\arrowreversed[very thick]{>}},
        mark=at position 0.75 with {\arrowreversed[very thick]{>}}},
              ]
	\draw[very thick,postaction={decorate}] \boundellipse{1.89,0}{2.5}{1.5}; ;
	\end{scope}
	\node at (4.1,1.3) {\tiny{$\Gamma_2(E)$}};
	\node at (-3,1.3) {\tiny{$\Gamma_1(E)$}};
	\node at (0.8,1.65) {\tiny{$\sqrt{E}$}};
	\node at (0.93,-1.64) {\tiny{$-\sqrt{E}$}};
	\node at (0.51,1.24) {\tiny{$\bullet$}};
        	\node at (0.51,-1.24) {\tiny{$\bullet$}};
	 \node at (-3.4,0) {\tiny{$\bullet$}};
	  \node at (4.4,0) {\tiny{$\bullet$}};
	  \node at (1.6,0) {\tiny{$\bullet$}};
	  \node at (-0.6,0) {\tiny{$\bullet$}};
	  \node at (0.35,-0.17) {\tiny{$0$}};
         \node at (-3.85,-0.17) {\tiny{$\alpha_1(E)$}};
          \node at (-1.1,-0.17) {\tiny{$\alpha_2(E)$}};
         \node at (2,-0.17) {\tiny{$\beta_1(E)$}};
         \node at (4.8,-0.17) {\tiny{$\beta_2(E)$}};
	\end{tikzpicture}
	\caption{Characteristic sets} \label{Fig2}
	\end{figure}
\end{center}

\noindent
{\bf Assumption (A2).} The interaction operator $W$ is a first order differential operator of the form
$$
W =r_0(x)+ir_1(x)hD_x, \quad D_x:= -i \frac{d}{dx},
$$
where $r_0, r_1$ are smooth real-valued bounded functions on $\mathbb R$, and satisfy the ellipticity condition at the crossing points $\rho_{\pm}(E)$,
\begin{equation}\label{Ell}
(r_0(0),r_1(0)) \not=(0,0).
\end{equation}

Under the above assumptions the operator $P(h)$ is self-adjoint in $L^2(\mathbb R)\oplus L^2(\mathbb R)$, and its spectrum in $I_0$ is at most discrete. We will consider the eigenvalue problem 
\begin{equation}\label{System}
P(h) u = E u, \;\; E=E(h)\in I_0 .
\end{equation}

We define the action integrals
\begin{equation}\label{actions}
{\mathcal A}_j(E):=  \int_{\alpha_j(E)}^{\beta_j(E)}\sqrt{ E-V_j(t)} \, dt \quad (j=1,2).
\end{equation}
The functions ${\mathcal A}_1(E),{\mathcal A}_2(E)$ are analytic and ${\mathcal A}'_1(E),{\mathcal A}'_2(E)$ are positive in $I_0$.

\begin{thm}\label{thQCS}
Assume \textbf{(A1)} and \textbf{(A2)}. There exist symbols $m_0(E;h), m_1(E;h),m_2(E;h)$ analytic with respect to $E\in I_0$, with 
\begin{equation}\label{smestimates}
m_j(E;h) = \mathcal{O}(h) \quad  (j=0,1,2),
\end{equation}
uniformly for $E\in I_0$ and $h>0$ small enough, such that  $E=E(h) \in I_0$ is an eigenvalue of $P(h)$ if and only if
\begin{equation}\label{precise quantification condition}
\left(\cos \left(\frac{\mathcal{A}_1(E)}{h} \right)  + m_1(E;h) \right) \left(\cos \left(\frac{\mathcal{A}_2(E)}{h}\right)+ m_2(E;h)\right) = m_0(E;h).
\end{equation}

In particular, in the symmetric case $V_1(x)=V_2(-x)$, we have
\begin{equation}\label{coefm0}
m_{0}(E;h) =\mathcal{D}(E) h + \mathcal{O}(h^2\ln(1/h)), 
\end{equation}
uniformly for $E\in I_0$ and $h>0$ small enough, where 
\begin{equation}\label{FuncD}
\mathcal{D}(E):= \frac{\pi}{2V_1'(0)} \left\vert r_0(0) E^{-\frac{1}{4}}\sin\left( \frac{\mathcal{B}(E)}{h}+ \frac{\pi}{4}\right)  + r_1(0) E^{\frac{1}{4}} \cos\left( \frac{\mathcal{B}(E)}{h}  +\frac{\pi}{4}  \right) \right\vert^2,
\end{equation}
and  $\mathcal{B}(E)$ is the action defined by 
\begin{equation}\label{defB}
 \mathcal{B}(E):=  2 \int_{0}^{\beta_1(E)} \sqrt{E-V_1(t)} dt.
\end{equation}
\end{thm}

This result entails the following one about the asymptotic distribution and the eigenvalue splitting of $P(h)$. We fix $E_0\in I_0$
and $h$-independent arbitrarily large constant $C_0>0$ and set $I_{h}:=[E_0-C_0h,E_0+C_0h]$.  We define an $h$-dependent discrete set $\mathcal{U}_h\subset I_h$ by
\begin{equation}\label{SetU}
\mathcal{U}_h:=\mathcal{U}^{(1)}_h\cup \mathcal{U}^{(2)}_h,
\quad
\mathcal{U}^{(j)}_h=\left \{E\in I_h; \exists \, k\in \Z\text{ s.t. }\mathcal{A}_j(E) = (k+\frac12)\pi h\right\},\quad j=1,2.
\end{equation}
For each $j=1,2$, the elements of $\mathcal{U}^{(j)}_h$ are simple roots of the Bohr-Sommerfeld quantization rule $\mathcal{A}_j(E) = (k+\frac12)\pi h$ for the scalar Schr\"odinger operator $P_j(h)$ with a simple potential well $(\alpha_j(E),\beta_j(E))$ and it is well known (see e.g. \cite{Ya} and references therein) that they approximate the eigenvalues of $P_j(h)$ in $I_0$ in the semiclassical limit $h\to 0^+$.
There may of course be intersection between the two sets $\mathcal{U}^{(1)}_h$ and $\mathcal{U}^{(2)}_h$.
In the following theorem, the elements of the set $\mathcal{U}_h$ should be counted repeatedly according to the multiplicity. Notice that in the symmetric case $V_1(x)=V_2(-x)$, the two sets $\mathcal{U}^{(1)}_h$ and $\mathcal{U}^{(2)}_h$ coincide, i.e., $\mathcal{U}^{(1)}_h = \mathcal{U}^{(2)}_h$.

\begin{thm}\label{mainth}
Assume \textbf{(A1)} and \textbf{(A2)}. There exists a subset $J\subset (0,1]$ satisfying $0\in\overline{J}$ such that for $h\in J$, there exists a bijection 
$$
b_h: \sigma\,(P(h))\cap I_{ h} \rightarrow \mathcal{U}_{h}
$$
such that 
$$
b_h(E) - E = \mathcal{O}(h^{\frac32}), 
$$
uniformly as $h\to 0^+$ in $J$.

In particular, in the symmetric case $V_1(x)=V_2(-x)$, the eigenvalues of $P(h)$ in $I_{h}$ come out  in pairs with splitting of order $h^{\frac32}$. More precisely, for $E\in \mathcal{U}_{h}$, if $E_+(h)$ and $E_-(h)$ are a pair of eigenvalues of $P(h)$ which are at a distance $\mathcal{O}(h^{\frac32})$ from $E$, then we have
$$
\vert E_+(h) - E_-(h)\vert = 2 \frac{\sqrt{\mathcal{D}(E)}}{\mathcal{A}'(E)} h^{\frac32} + \mathcal{O}\big(h^{\frac74}\big),
$$
uniformly as $h\to 0^+$ in $J$, where $\mathcal{A}(E):= \mathcal{A}_1(E)=  \mathcal{A}_2(E)$.
\end{thm}


\begin{remark}
The eigenvalue splitting occurs in fact for a pair of potentials satisfying $V_1(x)=V_2(-x)$ only in the classically allowed region $[\alpha_1(E),\beta_2(E)]$ and our results hold true under this weaker assumption.  But in the analytic case this means the global symmetry.
\end{remark}

\begin{remark}
The microlocal ellipticity condition \eqref{Ell} on the interaction operator $W$ at the crossing points $\rho_{\pm}(E)$ 
is relevant for our result of the eigenvalue splitting. If it is not satisfied the principal term ${\mathcal D}(E)$ vanishes. 
But the results hold true without this condition. In fact we used this condition only for the reduction to the scalar microlocal normal form of \cite{Sj1, CdvPa} (see the proof of Theorem \ref{TMF}), but the reduction to the Landau-Zener microlocal normal form due to Colin de Verdi\`ere \cite{Cdv2,Cdv3} is also possible without this condition.
\end{remark}

\section{Solutions to the system and existence of eigenvalues}
\label{sect3_0}

In this section we prove the existence of eigenvalues of the operator $P(h)$ in the interval $I_{0}$, together with a preliminary but a fundamental result on their location. The main idea relies on the construction of four $L^2$-solutions to the system \eqref{System}, two on each half-line $\mathbb R_-=(-\infty, 0]$ and $\mathbb R_+= [0,+\infty)$, 
$$
w_{1,L},w_{2,L} \in L^2(\mathbb R_-)\oplus L^2(\mathbb R_-), \quad w_{1,R},w_{2,R} \in L^2(\mathbb R_+)\oplus L^2(\mathbb R_+).
$$
The quantization condition that determines the eigenvalues of $P(h)$ in $I_{0}$ will then be given by the linear dependence condition $\mathcal{W}(E;h) = 0$, where $\mathcal{W}(E;h)$ stands for the wronskian of these solutions. The construction of such solutions in the case of a system similar to \eqref{System0} was carried out in \cite{FMW1, FMW3} starting from solutions to the underlying scalar equations $(P_j(h)-E)u=0$, $j=1,2$, with suitable asymptotic behaviors at infinity (see also \cite{Ya}). In the following, we simply recall the main lines of this construction together with some estimates important for our next purposes and we refer to these works for the details. 

\subsection{Solutions to the system} For $k\in \mathbb N$ and $U\subset \mathbb R$, we introduce the subspace $C^k_b(U)$ of $C^k(U)$ defined by 
$$
C^k_b(U) := \big\{ u\in C^k(U;\mathbb C) \,;\, \Vert u \Vert_{C^k_b(U)}:= \sum_{0\leq j\leq k} \sup_{x\in U} \vert u^{(j)}(x) \vert < +\infty \big\}.
$$

Let $E\in I_{0}$. For $j=1,2$, let $u_{j,L}^{\pm}$ be the solutions to the scalar equation 
$$
(P_j(h) - E) u = 0 \quad {\rm on} \;\; \mathbb R_-, 
$$
constructed in \cite[Appendix 2]{FMW1}. In particular, $u_{j,L}^-$ decays exponentially at $-\infty$, while $u_{j,L}^+$ grows exponentially, and their
Wronskian satisfies 
\begin{equation}\label{wronsL}
\W_{j,L}:=\W[u_{j,L}^-, u_{j,L}^+] = -\frac{2}{\pi } h^{-\frac23}(1+\ord(h)) \;\;\; \text{as}\;\; h \to 0^+.
\end{equation}
We introduce the linear operator
$$
K_{j,L}\, :\, C^0_b(\mathbb R_-) \to C^2_b(\mathbb R_-)
$$
defined by
\begin{equation*}
\label{eq1}
K_{j,L}[v](x) := \frac{1}{h^2 \W_{j,L}}\left( u_{j,L}^+(x) \int_{-\infty}^x  u_{j,L}^-(t)v(t)\,dt + u_{j,L}^-(x) \int_x^{0}\!\!\!\! u_{j,L}^+(t)v(t)\,dt\right), \;\;\; v\in C^0_b(\mathbb R_-) .
\end{equation*}
The operator $K_{j,L}$ is a fundamental solution of $(P_j(h)-E)u=0$ on $\mathbb R_-$, i.e., 
$$
(P_j(h)-E)K_{j,L}={\mathbf 1} \quad \text{on} \; \; \mathbb R_-, 
$$
and because of the form of the operator $W$, an integration by parts shows that we have,
$$
K_{j,L}W,\,\, K_{j,L}W^*\, :\, C^0_b(\mathbb R_-) \to C^0_b(\mathbb R_-) \quad (j=1,2).
$$

In a similar way, starting from solutions $u_{j,R}^{\pm}$ to the scalar equation $(P_j(h)-E) u =0$ on $\mathbb R_+$, the same construction holds and leads to a fundamental solution
$$
K_{j,R}\, :\, C^0_b(\mathbb R_+) \to C^2_b(\mathbb R_+), 
$$
of $P_j(h) - E$ on $\mathbb R_+$, i.e., 
$$
(P_j(h)-E)K_{j,R}={\mathbf 1} \quad \text{on} \; \; \mathbb R_+,
$$ 
and we have 
$$
K_{j,R}W,\,\, K_{j,R}W^*\, :\, C^0_b(\mathbb R_+) \to C^0_b(\mathbb R_+) \quad (j=1,2).
$$

Set 
$$
M_L:= h^2K_{1,L}WK_{2,L}W^*\in {\mathcal L}(C^0_b(\mathbb R_-)) , \quad N_L:= h^2K_{2,L}W^*K_{1,L}W \in {\mathcal L}(C^0_b(\mathbb R_-)),
$$
$$
M_R:= h^2K_{2,R}W^* K_{1,R}W \in {\mathcal L}(C^0_b(\mathbb R_+)),  \quad N_R:= h^2K_{1,R}WK_{2,R} W^*\in {\mathcal L}(C^0_b(\mathbb R_+)). 
$$

In view of the construction of solutions to the system \eqref{System}, the key result is the following proposition (see \cite[Proposition 3.1]{FMW3}).
\begin{proposition}
\label{PropK1K2.} As $h\rightarrow 0^+$,
one has,
\begin{equation}
\label{estnormM1L}
\Vert M_L\Vert_{{\mathcal L}(C^0_b(\mathbb R_-))} + \Vert N_L \Vert_{{\mathcal L}(C^0_b(\mathbb R_-))}= \ord (h^{\frac1{3}}),
\end{equation}
\begin{equation}
\label{estnormM1r}
\Vert M_R\Vert_{{\mathcal L}(C^0_b(\mathbb R_+))} + \Vert N_R \Vert_{{\mathcal L}(C^0_b(\mathbb R_+))}= \ord (h^{\frac1{3}}),
\end{equation}
\begin{equation}
\label{estnormK2L}
 |hK_{2,L}W^* v(0)| + |hK_{1,L}W v(0)| = \ord (\sup_{\mathbb R_-}|v|), \quad \forall v\in C^0_b(\mathbb R_-),
\end{equation}
\begin{equation}
\label{estnormK2L}
 |hK_{1,R}W v(0)| + |hK_{2,R}W^* v(0)| = \ord (\sup_{\mathbb R_+}|v|), \quad \forall v\in C^0_b(\mathbb R_+).
\end{equation}
\end{proposition} 

Now, using the fundamental solutions of the scalar equations $(P_j(h)-E)u=0$, $j=1,2$, constructed above, together with the above estimates, we can construct two solutions $w_{1,L},w_{2,L}$ to the system \eqref{System} on $\mathbb R_-$ which approach ${}^t(u_{1,L}^- , 0)$ and ${}^t(0, u_{2,L}^-)$, respectively,  and two other solutions $w_{1,R},w_{2,R}$ on $\mathbb R_+$ which approach ${}^t(u_{1,R}^- , 0)$ and ${}^t(0, u_{2,R}^-)$, respectively, as $h\to 0^+$.

On $\mathbb R_-$, setting $u={}^t(u_1, u_2)$, the system \eqref{System} can be re-written as
\begin{equation}\label{rew}
\left\{ \begin{array}{lll}
(P_1 - E) u_1 = - h W u_2 \\
(P_2-E)u_2 = -h W^* u_1.
\end{array}\right.
\end{equation}
Assume that $u_1\in C_b^0(\mathbb R_-)$ and set $u_2= -h K_{2,L}W^* u_1$. Then, \eqref{rew} reduces to the scalar equation of $u_1$
$$
(P_1 - E) u_1 = h^2 W K_{2,L} W^* u_1,
$$
and a solution will be given by any $u_1\in C_b^0(\mathbb R_-)$ such that 
$$
u_1 = u_{1,L}^- + M_L u_1.
$$
By Proposition \ref{PropK1K2.}, for $h>0$ small enough, the operator $({\mathbf 1}- M_L)^{-1}=\sum_{j\geq 0} M_L^j$ exists as an operator from $C^0_b(\mathbb R_-)$ to itself. Since by construction $u_{1,L}^-\in C^0_b(\mathbb R_-)$, we can define 
\begin{equation}\label{sol1}
w_{1,L} :=  \begin{pmatrix} \sum_{j\geq 0} M_L^j u_{1,L}^-  \\ - hK_{2,L}W^* \sum_{j\geq 0} M_L^j u_{1,L}^- \end{pmatrix} \in C_b^{0}(\mathbb R_-) \oplus C_b^{0}(\mathbb R_-),
\end{equation}
and we see that $w_{1,L}$ is solution to the system \eqref{System} on $\mathbb R_-$, with $w_{1,L} \rightarrow  \begin{pmatrix} u_{1,L}^-  \\ 0 \end{pmatrix}$ as $h\rightarrow 0^+$.

Similarly, $w_{2,L}$ defined by 
 \begin{equation}\label{sol2}
w_{2,L} :=  \left(\begin{array}{c} -h K_{1,L}W \sum_{j\geq 0} N_L^j u_{2,L}^-  \\ \sum_{j\geq 0} N_L^j u_{2,L}^- \end{array}\right)\in C_b^{0}(\mathbb R_-) \oplus C_b^{0}(\mathbb R_-)
\end{equation}
is solution to the system \eqref{System} on $\mathbb R_-$, with $w_{2,L} \rightarrow  \begin{pmatrix} 0\\ u_{2,L}^-   \end{pmatrix}$ as $h\rightarrow 0^+$.

On $\mathbb R_+$, a similar construction can be done and we see that the convergent series 
\begin{equation}\label{sol3}
w_{1,R}:=   \left(\begin{array}{c} \sum_{j\geq 0} N_R^j u_{1,R}^-  \\ -hK_{2,R}W^*\sum_{j\geq 0} N_R^j u_{1,R}^- \end{array}\right)\in C_b^{0}(\mathbb R_+) \oplus C_b^{0}(\mathbb R_+),
\end{equation}
\begin{equation}\label{sol4}
w_{2,R}:=   \left(\begin{array}{c} -hK_{1,R} W\sum_{j\geq 0} M_R^j u_{2,R}^-  \\ \sum_{j\geq 0} M_R^j u_{2,R}^- \end{array}\right)\in C_b^{0}(\mathbb R_+) \oplus C_b^{0}(\mathbb R_+),
\end{equation}
are both solutions to \eqref{System} on $\mathbb R_+$ and they tend respectively to $\begin{pmatrix} u_{1,R}^-\\0\end{pmatrix}$ and $\begin{pmatrix} 0 \\u_{2,R}^-\end{pmatrix}$ as $h \rightarrow 0^+$.

The constructed solutions are actually in $L^2$, more precisely we have (see \cite[Proposition 4.1]{FMW1}),

\begin{proposition} The solutions $w_{j,L}$ given by \eqref{sol1}-\eqref{sol2}, and $w_{j,R}$ given by \eqref{sol3}-\eqref{sol4}, satisfy
$$
w_{j,L} \in L^2(\mathbb R_-) \oplus L^2(\mathbb R_-), \quad w_{j,R} \in L^2(\mathbb R_+) \oplus L^2(\mathbb R_+) \quad (j=1,2).
$$
\end{proposition}

\subsection{Quantization condition and existence of eigenvalues} 

Using the four solutions to the system \eqref{System} constructed in the previous paragraph, we can now deduce the quantization condition that determines the eigenvalues of $P(h)$ in $I_0$. In fact, by the general theory on systems of ordinary differential equations, the subspace of solutions of \eqref{System} that are in $L^2(\mathbb R_-)\oplus L^2(\mathbb R_-)$ is of dimension $2$, and the same is true for the subspace of solutions that are in $L^2(\mathbb R_+)\oplus L^2(\mathbb R_+)$. Therefore, $E=E(h)\in I_{0}$ is an eigenvalue of $P(h)$ if and only if the solutions $w_{1,L}, w_{2,L}, w_{1,R}, w_{2,R}$ are linearly dependent, which equivalent to say that
\begin{equation}
\label{condquant}
 {\mathcal W}(E;h)=0,
\end{equation}
where ${\mathcal W}(E;h):={\mathcal W}(w_{1,L}, w_{2,L}, w_{1,R}, w_{2,R})$ stands for the Wronskian of $w_{1,L}, w_{2,L}, w_{1,R}$ and $w_{2,R}$, that is 
$$
 {\mathcal W}(E;h) := \det \begin{pmatrix}
 w_{1,L} & w_{2,L} &  w_{1,R} & w_{2,R} \\
 w'_{1,L} & w'_{2,L}&  w'_{1,R} & w'_{2,R}
\end{pmatrix}.
$$
\begin{proposition}
\label{approxW0}
We have 
$$
{\mathcal W}(E;h)=\frac {16}{\pi^2}h^{-\frac43}\cos \left(\frac{{\mathcal A_1}(E)}{h} \right) \cos \left(\frac{{\mathcal A_2}(E)}{h}\right)+ {\mathcal O}(h^{-\frac76}),
$$
uniformly for $E\in I_{0}$ and $h>0$ small enough. 
\end{proposition}
\begin{proof} 
Since ${\mathcal W}(E;h)$ is constant with respect to $x$, it is enough to compute it at a fixed point, say at $x=0$. We recall the following estimates from \cite{FMW1}. For $S=L,R$, and $h>0$ small enough, one has,
$$
\begin{aligned}
& w_{1,S}(0) = \begin{pmatrix} u_{1,S}^-(0)\\\\ 0\end{pmatrix}+\cO(h^{\frac13})\quad ; \quad w_{1,S}'(0)=  \begin{pmatrix} (u_{1,S}^-)'(0)\\\\ 0\end{pmatrix}+\cO(h^{-\frac23}); \\
& w_{2,S}(0) =  \begin{pmatrix} 0 \\ \\ u_{2,S}^-(0) \end{pmatrix}+\cO(h^{\frac13})\quad ; \quad w_{2,S}'(0)=  \begin{pmatrix} 0 \\ \\ (u_{2,S}^-)'(0)  \end{pmatrix} +\cO(h^{-\frac23}).
\end{aligned}
$$
Moreover, we have, uniformly as $h\rightarrow 0^+$,
$$
u_{j,S}^-(0)=\mathcal{O}(h^{\frac16}) \;\;\; ;\;\;\ (u_{j,S}^-)'(0)=\mathcal{O}(h^{-\frac56}).
$$
Using the above estimates, we immediately obtain, for $h>0$ small enough,
$$
{\mathcal W}(E;h)={\mathcal W}(u_{1,L}^-,u_{1,R}^-){\mathcal W}(u_{2,L}^-,u_{2,R}^-)+\cO(h^{-\frac76}),
$$
where ${\mathcal W}(u_{j,L}^-,u_{j,R}^-)$ stands for the Wronskian of $u_{j,L}^-,u_{j,R}^-$, $j=1,2$. On the other hand, we know from standard WKB constructions (see also \cite[Appendix]{FMW1}) that we have,
\begin{equation}
 {\mathcal W}(u_{j,L}^-,u_{j,R}^-)= -\frac4{\pi}h^{-\frac23}\cos \left(\frac{{\mathcal A}_j(E)}{h} \right)+ \cO(h^{\frac13}), \quad j=1,2.
\end{equation}
This ends the proof of the proposition.
\end{proof}

Now, we are able to establish the existence of eigenvalues, together with a preliminary (but fundamental) result on their location. Let $E_0\in I_0$ and set $I_{h}:=[E_0-C_0h, E_0+C_0h]$ with $C_0>0$ fixed arbitrarily large and $h$-independent. With the definition of $\mathcal{U}_h$ given by \eqref{SetU}, we have
\begin{thm} \label{ThapproxRes}
Assume \textbf{(A1)} and \textbf{(A2)}. There exists a subset $J\subset (0,1]$ satisfying $0\in \overline{J}$ such that for all $h\in J$, there exists a bijection 
$$
b_h: \sigma\,(P(h))\cap I_{h} \rightarrow \mathcal{U}_{h}
$$
such that 
\begin{equation}\label{estbij}
b_h(E) - E = \mathcal{O}(h^{\frac76}), 
\end{equation}
uniformly as $h\to 0^+$ in $J$.
\end{thm}
\begin{proof} We set
$$
\mathcal{G}(E,h):=\mathcal{F}(E;h) -\frac{\pi^2h^{\frac43}}{16}{\mathcal W}(E;h), \;\;\; \mathcal{F}(E;h):= \cos \left(\frac{\mathcal{A}_1(E)}{h} \right) \cos \left(\frac{\mathcal{A}_2(E)}{h}\right).
$$
Then, the quantization condition \eqref{condquant} can be written as
\begin{equation}\label{condquantbis}
  \mathcal{F}(E;h) =\mathcal{G}(E,h).
\end{equation}
By Proposition \ref{approxW0}, we have $\mathcal{G}(E,h)=\cO(h^{\frac16})$, uniformly for $E\in I_{0}$ and $h>0$ small enough. To solve equation \eqref{condquantbis}, we first observe that in $I_{h}$, the roots of the equation $\mathcal{F}(E;h)=0$ are precisely the set of approximate eigenvalues $\mathcal{U}_h =\mathcal{U}_h^{(1)} \cup \mathcal{U}_h^{(2)}$ defined by \eqref{SetU}. For $h>0$ small enough and $e(h)\in \mathcal{U}_h$, there exist $j\in\{1,2\}$ and a unique $k\in \mathbb Z$ such that $\mathcal{A}_j(e(h))=(k+\frac12)h\pi$, and by Taylor's formula, we have 
\begin{equation}\label{Tapr}
e(h) = E_0 + h \bigg(\frac{\pi}{\mathcal{A}'_j(E_0)} (k+\frac12) - \frac{\mathcal{A}_j(E_0)}{\mathcal{A}'_j(E_0)} h^{-1} + \mathcal{O}(h)\bigg). 
\end{equation}
We recall that 
$$
{\mathcal A}_j'(E_0)=\frac12\int_{\alpha_j(E_0)}^{\beta_j(E_0)}(E_0-V_j(x))^{-1/2}dx\not=0\quad (j=1,2).
$$ 
In particular, the cardinal of $\mathcal{U}_h$ is uniformly bounded with respect to $h$, for $h>0$ small enough. 

Now, using \eqref{Tapr}, it is easy to check that, for all $C_1>0$, there exist $C_2>0$ and a subset $J\subset (0,1]$ with $0\in \overline{J}$ such that for all $h\in J$ and all $e(h)\in \mathcal{U}_h$, we have $\left\{z\in \mathbb C;\, \vert z-e(h)\vert\leq C_1 h^{\frac76} \right\}\cap \mathbb R \subset I_{h}$, and 
$$
\left\{z\in \mathbb C;\, C_2h^2\leq \vert z-e(h)\vert\leq C_1 h^{\frac76} \right\} \cap \mathcal{U}_h=\emptyset.
$$
Moreover, the set $\left\{z\in \mathbb C;\, \vert z-e(h)\vert\leq C_2 h^{2} \right\} \cap \mathcal{U}_h$ contains at most two elements including $e(h)$. Thus, we can apply Rouch\'e Theorem on $\left\{z\in \mathbb C;\, \vert z-e(h)\vert\leq \frac{C_1}{2} h^{\frac76} \right\}$, using that
$$
\vert \mathcal{F}(z;h) \vert \geq C' h^{\frac16} ,\;\; \forall z\in \left\{z\in \mathbb C;\, \vert z-e(h)\vert= \frac{C_1}{2} h^{\frac76} \right\},
$$
where $C'>0$ can be taken arbitrarily large by taking $C_1>0$ large enough, and conclude that there exists a bijection $b_h$ between the set of zeros of \eqref{condquantbis} in $I_{h}$ which are exactly the eigenvalues of $P(h)$ in $I_{h}$ and the set $\mathcal{U}_{h}$. Furthermore, this bijection satisfies \eqref{estbij} in the semiclassical limit $h\to 0^+$. This ends the proof of the Theorem.
\end{proof}

\section{Microlocal study of solutions at a crossing point for a general system}
Now, in order to improve the quantization rule that determines the eigenvalues of $P(h)$ in $I_0$ and to obtain a better estimate on their location, we use a microlocal approach that relies on the study of the behavior of the corresponding eigenfunctions microlocally near the characteristic set 
\begin{equation}\label{Characteristic set}
{\rm Char} (P(h)-E)=\Gamma_1(E)\cup \Gamma_2(E).
\end{equation}
The key point in this method consists on the computation of the connection formulae that link the microlocal data at the crossing and turning points of ${\rm Char} (P(h)-E)$. In this section, we study the microlocal solutions of a general $2\times 2$ system of pseudodifferential operators near a transversal crossing point of its charactersitic set and we compute the transfer matrix at this point.

\subsection{Semiclassical and microlocal terminologies}  Let us recall briefly some basic notions of semiclassical and microlocal analysis, referring to the books \cite{DiSj, Ma,Zw} for more details. We will use the notations of \cite{DiSj} for symbols and $h$-pseudodifferential operators. In particular, $S^0$ is the space of symbols 
$$
S^0:= \left\{q\in C^{\infty}(\mathbb R^2;\mathbb C); \, \big\vert\partial_x^{\alpha}\partial_{\xi}^{\beta} q(x,\xi;h) \big\vert = \mathcal{O}_{\alpha,\beta}(1), \;  \forall \alpha,\beta\in \mathbb N \right\},
$$
uniformly with respect to $h>0$ small enough. We recall that for a symbol $q\in S^0$, the corresponding $h$-pseudodifferential operator denoted ${\rm Op}_h^w(q)$ can be defined using the $h$-Weyl quantization by 
$$
{\rm Op}_h^w(q) u(x) := \frac{1}{2\pi h} \int_{\mathbb R^2} e^{i(x-y)\xi/h} q\left(\frac{x+y}{2},\xi;h\right) u(y) dy d\xi, \;\;\; u\in C_0^{\infty}(\mathbb R).
$$

Since our study is of microlocal nature and since we shall constantly use this vocabulary in the following, we briefly recall from \cite{Ma} (see also \cite{HeSj3}) the meaning of expression like $u=0$ microlocally in some open subset of the phase space. For $u\in \mathcal{S}'(\mathbb R)$, we denote $\mathcal{T}u$ the so-called semiclassical FBI-transform of $u$ given by 
$$
\mathcal{T}u(x,\xi;h) := 2^{-1/2} (\pi h)^{-3/4} \int_{\mathbb R} e^{i(x-y)\xi/h - (x-y)^2/2h} u(y) dy.
$$
The function $\mathcal{T}u$ is a $C^{\infty}$ function on $\mathbb R^2$. Let $\Omega\subset \mathbb R^2$ and $u=u(x;h)\in L^2(\mathbb R)$ with $\Vert u\Vert_{L^2}\leq 1$. We say that $u$ is microlocally $0$ in $\Omega$ and we write 
$$
u\sim 0 \;\; {\rm microlocally \; in} \; \Omega,
$$ 
if and only if $\Vert \mathcal{T} u\Vert_{L^2(\Omega)} = \mathcal{O}(h^{\infty})$. Here we use the standard asymptotic notation $f_h=\mathcal{O}(h^{\infty})$ which means that $f_h=\mathcal{O}(h^k)$ for all $k\in \mathbb N$ and $h>0$ small enough. The closed set of points where $u$ is not microlocally $0$ is called the frequency set of $u$ and denoted ${\rm FS}\,(u)$. This notion is analogous to the notion of microsupport in the analytic framework (see \cite{Ma,HeSj3}). For a $h$-pseudodifferential operator $A(h)$, we say that $u$ is a microlocal solution to the equation $A(h)u=0$ in $\Omega$ and we write 
$$
A(h) u \sim 0 \;\; {\rm microlocally \; in} \;\; \Omega,
$$ 
if $\Omega\cap {\rm FS}\,(A(h)u)=\emptyset$.


Consider a $2\times 2$ system 
\begin{equation}
\mathcal{Q}:= \begin{pmatrix} \mathcal{Q}_1 & h \mathcal{R}\\\\ h \mathcal{R}^* & \mathcal{Q}_2\end{pmatrix},
\end{equation}
where $\mathcal{Q}_1,\mathcal{Q}_2$ and $\mathcal{R}$ are pseudodifferential operators with symbols $q_1(x,\xi)$, $q_2(x,\xi) $ and $r(x,\xi)$ respectively. We make the following assumptions.

\begin{itemize}

\item[i)] The symbols $q_1,q_2\in S^0$ are real-valued, vanish at $\rho_0=(0,0)$, i.e. 
$$
q_1(\rho_0)=q_2(\rho_0) =0,
$$ 
and satisfy the following conditions 
\begin{equation*}
\partial_{\xi} q_1(\rho_0) \partial_{\xi} q_2(\rho_0) >0, 
\end{equation*}
\begin{equation}\label{Cond model}
\{q_1,q_2\}(\rho_0)>0,
\end{equation}
where $\{q_1,q_2\}(x,\xi):=(\partial_{\xi}q_1 \partial_x q_2 - \partial_x q_1 \partial_{\xi}q_2)(x,\xi)$ denotes the Poisson bracket of $q_1,q_2$. 

\vspace{0.2cm}

\item[ii)] The symbol $r\in S^0$ of the interaction operator $\mathcal{R}$ satisfies the ellipticity condition at $\rho_0$,
\begin{equation}\label{Ellipticity condition}
r(\rho_0)\neq 0.
\end{equation}
\end{itemize}

In particular, the condition \eqref{Cond model} means that the characteristic sets 
$$
\Gamma_{q_j}:= \{(x,\xi)\in \mathbb R^2; q_j(x,\xi)=0\}\;\;\; (j=1,2).
$$ 
intersect transversally at $\rho_0$. We plan to study the microlocal solutions to the system
\begin{equation}\label{GSystem}
\mathcal{Q}u=0,
\end{equation} 
 microlocally near the crossing point $\rho_0=(0,0)$.

\begin{comment}
\begin{center}
\begin{tikzpicture}
\centerarc[thick,red](0,0)(200:285:2.5cm);
\centerarc[thick,red](-2,-0.5)(-102:0:2cm);
\end{tikzpicture}
\end{center}
\end{comment}


\subsection{Microlocal WKB solutions} Since the operator $\mathcal{Q}$ is microlocally elliptic outside its characteristic set $\Gamma_{q_1}\cup \Gamma_{q_2}$, it follows by standard arguments of microlocal analysis that the solutions of the system \eqref{GSystem} are microlocally supported in a neighborhood of this set. First, we study these microlocal solutions away from the crossing point $\rho_0$, that is, near the four curves 
$$
\Gamma_{q_1}^+:= \left\{(x,\xi)\in \Gamma_{q_1};\, q_2(x,\xi)>0 \right\}, \quad \Gamma_{q_1}^-:= \{(x,\xi)\in \Gamma_{q_1}; \, q_2(x,\xi)<0\},
$$
$$
\Gamma_{q_2}^+:= \{(x,\xi)\in \Gamma_{q_2}; \,q_1(x,\xi)>0\}, \quad \Gamma_{q_2}^-:= \{(x,\xi)\in \Gamma_{q_2};\,  q_1(x,\xi)<0\}.
$$

On $\Gamma_{q_1}^{\pm}$, the operator $\mathcal{Q}_1$ is of real principal type while $\mathcal{Q}_2$ is elliptic, and the same is true on $\Gamma_{q_2}^{\pm}$ by interchanging $\mathcal{Q}_1$ and $\mathcal{Q}_2$. Hence microlocally on each of the four curves $(\Gamma_{q_j}^{\pm})_{j=1,2}$, the system \eqref{GSystem} is reduced  to a scalar one-dimensional equation. Thus, the space of microlocal solutions on each of these curves is one-dimensional and a basis of WKB solutions is given by the following proposition.
\begin{proposition}\label{GenWKB}
On each of the curves $\Gamma_{q_j}^{\pm}$, the space of microlocal solutions to the system \eqref{GSystem} is one-dimensional and there exist $f_{q_j}^{\pm}$ such that for $j=1,2$,
$$
\mathcal{Q} f_{q_j}^{\pm} \sim 0 \,\,\,{\rm microlocally\, on}\,\, \Gamma_{q_j}^{\pm},
$$
and $f_{q_j}^{\pm}$ have the following WKB form
\begin{equation}\label{WKB form}
f_{q_j}^{\pm}(x;h) \sim \begin{pmatrix}
a_{q_j}(x;h) \\\\
b_{q_j}(x;h)
\end{pmatrix}
e^{i\phi_{q_j}(x)/h}
\,\,\,{\rm microlocally\, on}\,\,
\Gamma_{q_j}^{\pm},
\end{equation}
where the phase function $\phi_{q_j}$ is defined as the unique solution of the eikonal equation
\begin{equation}\label{eikonalphase}
\left\{\begin{array}{lll}
q_j(x,\phi_{q_j}'(x))=0, \\
\phi_{q_j}(0) =0,
\end{array}\right.
\end{equation}
and $a_{q_j}(x;h)$, $b_{q_j}(x;h)$ are symbols of the form
$$
a_{q_j}(x;h)\sim \sum_{k\geq0}h^k a_{q_j,k}(x), \quad b_{q_j}(x;h)\sim \sum_{k\geq0}h^k b_{q_j,k}(x),
$$ 
with leading terms given by
\begin{equation*}
\begin{aligned}
&a_{q_1,0}(x) =  \exp\left( -\int_0^x \frac{ \partial_x\partial_{\xi}q_1(t,\phi_{q_1}'(t)) + \phi_{q_1}''(t)\partial_{\xi}^2 q_1(t,\phi_{q_1}'(t))}{ 2\partial_{\xi} q_1(t,\phi_{q_1}'(t))       } dt\right), \\
&b_{q_1,0}(x) = 0, \\
&b_{q_1,1}(x) = - \frac{ \overline{r(x,\phi_{q_1}'(x))}  }{q_2(x,\phi_{q_1}'(x))} a_{q_1,0}(x),
\end{aligned}
\end{equation*}
and 
\begin{equation*}
\begin{aligned}
&b_{q_2,0}(x) =  \exp\left( -\int_0^x \frac{ \partial_x\partial_{\xi}q_2(t,\phi_{q_2}'(t)) + \phi_{q_2}''(t)\partial_{\xi}^2 q_2(t,\phi_{q_2}'(t))     }{ 2\partial_{\xi} q_2(t,\phi_{q_2}'(t))       } dt\right), \\
&a_{q_2,0}(x) = 0, \\
&a_{q_2,1}(x) = - \frac{ \overline{r(x,\phi_{q_2}'(x))}  }{q_1(x,\phi_{q_2}'(x))} b_{q_2,0}(x).
\end{aligned}
\end{equation*}
\end{proposition}

\begin{remark}
In particular, we have the following asymptotic behaviors as $x\to 0$
$$
a_{q_1,0}(x) = 1+ \mathcal{O}(x) , \quad b_{q_1,1}(x) = - \frac{\partial_{\xi} q_1(\rho_0)\overline{r(\rho_0)}}{\{q_1,q_2\}(\rho_0)} \left( \frac{1+ \mathcal{O}(x)}{x}\right),
$$
$$
b_{q_2,0}(x) =  1+ \mathcal{O}(x) , \quad a_{q_2,1}(x) = \frac{\partial_{\xi} q_2(\rho_0)r(\rho_0)}{\{q_1,q_2\}(\rho_0)} \left(\frac{1+ \mathcal{O}(x)}{x}\right).
$$

\end{remark}

We refer to the Appendix \ref{MWSA} for the construction of these WKB solutions.

\subsection{Transfer matrix} Now, we state the main result of this section which provides the transfer formula of microlocal solutions near the crossing point $\rho_0$.

\begin{theorem}\label{TMF}
Let $u(x;h)\in L^2(\R)$ be a solution to the system $\mathcal{Q} u\sim 0$ microlocally in a small neighbourhood of $\rho_0=(0,0)$ such that 
$$
u \sim t_{j}^{\pm} f_{q_j}^{\pm} \,\,\,{\rm microlocally\, on}\,\,\Gamma_{q_j}^{\pm}, 
$$
for some scalar complex numbers $t_{j}^{\pm}=t_{j}^{\pm}(h)$, $j=1,2$. Then, there exist classical symbols of order $0$, $\mu=\mu(h)\sim \sum_{k\geq 0}h^k \mu_{k}$ and $\widehat{\mu}=\widehat{\mu}(h)\sim \sum_{k\geq 0}h^k \widehat{\mu}_{k}$ such that
\begin{equation}
\label{tau-}
\begin{pmatrix}
t_{1}^{+} \\\\
t_{2}^{+}
\end{pmatrix}
= \begin{pmatrix}
\kappa_{1,1}(h) & h^{\frac12-ih\mu }\kappa_{1,2}(h) \\\\
h^{\frac12-ih\widehat{\mu}} \kappa_{2,1}(h) & \kappa_{2,2}(h)
\end{pmatrix}
\begin{pmatrix}
t_{1}^{-} \\\\
t_{2}^{-}
\end{pmatrix},
\end{equation}
where $\kappa_{j,k}(h)\sim \sum_{n\geq 0} h^n \kappa_{j,k}^n$ are symbols with leading terms given by 
$$
\kappa_{1,1}^0 = \kappa_{2,2}^0=1,
$$
$$
\kappa_{1,2}^0 = - e^{i\frac{\pi}{4}} \left(r(x,\xi) \sqrt{\frac{2\pi \partial_{\xi} q_2(x,\xi)}{\partial_{\xi}q_1(x,\xi)\{q_1,q_2\}(x,\xi)}}\right)_{\vert (x,\xi)=\rho_0},
$$
$$
\kappa_{2,1}^0 = - e^{-i\frac{\pi}{4}} \left(\overline{r(x,\xi)} \sqrt{\frac{2\pi \partial_{\xi} q_1(x,\xi)}{\partial_{\xi}q_2(x,\xi)\{q_1,q_2\}(x,\xi)}}\right)_{\vert (x,\xi)=\rho_0}.
$$
\end{theorem}

The rest of this section is devoted to the proof of this result which relies on several steps. The first step consists to reduce the system \eqref{GSystem} to a scalar equation using the ellipticity condition \eqref{Ellipticity condition} and then to solve this equation by means of a normal form in the spirit of \cite{HeSj3, Sj1,CdvPa}. 

To simplify the notations, we set 
\begin{equation}\label{Notationsabcd}
\alpha:= \partial_{x} q_1(\rho_0), \; \beta:=  \partial_{\xi} q_1(\rho_0), \; \gamma:= \partial_{x} q_2(\rho_0), \;  \delta:=\partial_{\xi} q_2(\rho_0), \; D:=\{q_1,q_2\}(\rho_0).
\end{equation}

\subsubsection{Reduction to a scalar equation and normal form}\label{RSE}

Setting $u={}^t(u_1,u_2)$ and using the ellipticity of $\mathcal{R}$ at $\rho_0$ according to assumption \eqref{Ellipticity condition}, the system \eqref{GSystem} is reduced microlocally near the origin to a scalar equation of $u_1$. More precisely, there exists a small neighborhood $\mathcal{V}\subset \mathbb R^2$ of $\rho_0$ such that microlocally in $\mathcal{V}$, the system $\mathcal{Q} u \sim 0$ is reduced to 
\begin{equation}\label{FR1}
\left\{\begin{array}{lll}
\mathcal{L}u_1 \sim 0,\\
u_2 \sim - h^{-1} \mathcal{R}^{-1} \mathcal{Q}_1 u_1,
\end{array}\right.
\end{equation}
where $\mathcal{R}^{-1}$ denotes a parametrix of $\mathcal{R}$ in $\mathcal{V}$ and $\mathcal{L}$ is the $h$-pseudodifferential operator defined by 
$$
\mathcal{L} := \mathcal{R} \mathcal{Q}_2 \mathcal{R}^{-1} \mathcal{Q}_1 - h^2 \mathcal{R} \mathcal{R}^*.
$$
In particular, the principal symbol of $\mathcal{L}$ is given by $\ell_0(x,\xi):= q_1(x,\xi) q_2(x,\xi)$. The crossing point $\rho_0$ is a saddle point for $\ell_0$, and a normal form for $\mathcal{L}$, reducing its study microlocally near $\rho_0$ to that of the operator $\frac{1}{2} (yh D_y + h D_y\cdot y)$ microlocally near the origin is well-known (see \cite{Sj1,CdvPa} in the $C^{\infty}$ case and \cite{HeSj3} in the analytic case). More precisely, we have

\begin{lemma}
There exist a small neighborhood $\Omega\subset \mathbb R^2$ of $(0,0)$, a Fourier integral operator $U$, not necessary unitary, with associated canonical transformation $\kappa$ sending $\mathcal{V}$ to $\Omega$ and $\kappa(\rho_0)=(0,0)$, and a classical symbol $F(t;h)\sim \sum_{k\geq 0} h^k F_k(t)\in C^{\infty}$ defined near $t=0$, such that 
\begin{equation}\label{FR}
U F(\mathcal{L};h) U^{-1} \sim \mathcal{G}:= \frac{1}{2} (yh D_y + h D_y\cdot y)  \;\;\; {\rm microlocally\, in} \;\; \Omega.
\end{equation}
Moreover, we can choose the reduction in such a way that we have 
\begin{equation}\label{FSF}
F(0;h) = -\frac{i}{2} h + \mu h^2,
\end{equation}
where $\mu=\mu(h)\sim \sum_{k\geq 0} h^k \mu_{k}$ is a classical symbol of order $0$.
\end{lemma}

\begin{proof}
The normal form \eqref{FR} is due to \cite{Sj1,CdvPa} in the smooth case. Notice that in these works, this result was proved for self-adjoint operators, but it still holds for our non-self-adjoint operator $\mathcal{L}$ thanks to the form of its principal symbol $\ell_0$.
In the following, we prove \eqref{FSF}.

The FIO $U$ is associated with the canonical transform $\kappa: (x,\xi)\mapsto (y,\eta)$ satisfying 
$$
F(\ell(x,\xi;h)) = y \eta,
$$
where $\ell(x,\xi;h)\sim \sum_{k\geq 0}h^k \ell_k(x,\xi)$ is the full symbol of $\mathcal{L}$. In particular, we can choose $\kappa(x,\xi)=\kappa_0(x,\xi) + \mathcal{O}((x,\xi)^2)$ with 
$$
\kappa_0(x,\xi) = \frac{1}{\sqrt{D}} \left( \gamma x + \delta \xi, \alpha x + \beta \xi\right).
$$

After a convenient normalization, we can write $U^{-1}$ in the form 
$$
U^{-1} v(x;h) = \int_{\mathbb R} e^{i\psi(x,y)/h} c(x,y;h) v(y) dy,
$$
where $c(x,y;h)\sim \sum_{k\geq 0} c_k(x,y)$ is a symbol with $c_0(0,0)=1$ and the phase function $\psi(x,y)$ is a generating function of $\kappa^{-1}$, in the sense that $\kappa^{-1}: (y,-\nabla_y \psi)\mapsto (x,\nabla_x \psi)$. In particular, near $(x,y)=\rho_0$, we have
\begin{equation}\label{POIF}
\psi(x,y)=\frac{1}{2\delta}(-\gamma x^2+ 2 \sqrt{D}xy- \beta y^2) + \mathcal{O}((x,y)^3).
\end{equation}
At the levels of principal and sub-principal symbols, the relation \eqref{FR} implies that 
$$
F_0(\ell_0(\kappa^{-1}(y,\eta))) = y \eta,
$$
$$
F_1(\ell_0(\kappa^{-1}(y,\eta))) + \ell_1(\kappa^{-1}(y,\eta)) F_0'(\ell_0(\kappa^{-1}(y,\eta))) = 0.
$$
In particular, the first equation at $(y,\eta)=(0,0)$ implies that $F_0(0)=0$ and $F_0'(0)= \frac{1}{D}$, and the second one gives 
$$
F_1(0)=-\ell_1(\rho_0)F_0'(0)=-\frac{i}{2},
$$
since $\ell_1(\rho_0)=\frac{iD}{2}$. Thus the symbol $F(0;h)$ has the form \eqref{FSF}.
\end{proof}

\subsubsection{Microlocal solutions near the crossing point} Setting $\tilde{u}_1:= U u_1$, the equation $\mathcal{L}u_1\sim 0$ microlocally in $\mathcal{V}$ is equivalent to 
\begin{equation}\label{Eqreduite}
\mathcal{G} \tilde{u}_1 \sim F(0;h) \tilde{u}_{1}\;\;\; {\rm microlocally \; in}\;\; \Omega,
\end{equation}
which can be rewritten as
$$
y (\tilde{u}_1)'(y) \sim i \mu h \tilde{u}_1(y) \;\;\; {\rm microlocally \; in}\;\; \Omega.
$$
The space of microlocal solutions of this equation is two dimensional and a basis is given by the two functions 
\begin{equation}\label{gfunctions}
g_{\mu}^+(y) := H(y) y^{i\mu h}, \quad g_{\mu}^-(y) := H(-y) \vert y \vert^{i\mu h},
\end{equation}
where $H$ denotes the Heaviside function, i.e., $H(y)=1$ for $y\geq 0$ and $H(y) = 0$ for $y<0$. In particular, we have 
$$
{\rm FS}( g_{\mu}^{\pm})=\{\pm y> 0, \eta=0\} \cup \{y=0\}.
$$
Thus, $u_1^{\pm} := U^{-1} g_{\mu}^{\pm}$ are solutions to the equation $\mathcal{L}u_1\sim 0$ microlocally in $\mathcal{V}$, and we have 
$$
{\rm FS}(u_1^{\pm}) \cap \mathcal{V} \subset  \big(\Gamma_{q_1}^{\pm} \cup \Gamma_{q_2}\big) \cap \mathcal{V}.
$$
More precisely, we have the following asymptotic formulae for $u_1^{\pm}$.
\begin{proposition}\label{pre1}
There exist symbols 
$$
\sigma^{\pm}(x;h)\sim \sum_{k\geq 0}h^k \sigma_k^{\pm}(x), \;\;\; \eta^{\pm}(x;h)\sim \sum_{k\geq 0}h^k \eta_k^{\pm}(x),
$$ 
with leading terms given by 
$$
\sigma_0^+(x) = \sqrt{\frac{\delta}{\beta}} e^{-i\frac{\pi}{4}} + \mathcal{O}(x) , \quad  \eta_0^+(x) = \frac{i \delta}{\sqrt{D}x}(1+\mathcal{O}(x)),
$$
$$
\sigma_0^-(x) = \sqrt{\frac{\delta}{\beta}} e^{-i\frac{\pi}{4}} + \mathcal{O}(x) , \quad  \eta_0^-(x) = - \frac{i \delta}{\sqrt{D}x}(1+\mathcal{O}(x)),
$$
such that, modulo $\mathcal{O}(h^{\infty})$ as $h\to 0^+$, we have 

$$
u_1^+(x;h) = \left\{\begin{array}{lll}
\sqrt{2\pi h}\, \sigma^+(x;h) e^{i\phi_{q_1}(x)/h} + h^{1+ i\mu h} \eta_+(x;h) e^{i\phi_{q_2}(x)/h}  \;\;\; & (x>0) \\ \\
 h^{1+ i\mu h} \, \eta_+(x;h) e^{i\phi_{q_2}(x)/h}  \;\;\; & (x<0)
\end{array}\right.
$$
and
$$
u_1^-(x;h) = \left\{\begin{array}{lll}
 h^{1+ i\mu h} \,\eta_-(x;h) e^{i\phi_{q_2}(x)/h}  \;\;\; & (x>0) \\ \\
\sqrt{2\pi h} \, \sigma^-(x;h) e^{i\phi_{q_1}(x)/h} + h^{1+ i\mu h} \eta_-(x;h) e^{i\phi_{q_2}(x)/h}  \;\;\; & (x<0).
\end{array}\right.
$$
\end{proposition}

\begin{proof}
We only prove the formula for $u_1^+$. The proof for the other one is similar. By definition, we have
\begin{equation}\label{SPO}
u_1^+(x;h) = U^{-1} g_{\mu}^+(x;h)= \int_0^{+\infty} e^{i\psi(x,y)/h} c(x,y;h) y^{i\mu h} dy,
\end{equation}
where we recall that the phase function $\psi$ satisfies (see \eqref{POIF})
$$
\psi(x,y)=\frac{1}{2\delta}(-\gamma x^2+ 2 \sqrt{D}xy-\beta y^2) + \mathcal{O}((x,y)^3) \;\;\; {\rm as}\;\; (x,y)\to \rho_0.
$$
The right hand side of \eqref{SPO} is an oscillatory integral and up to terms of order $\mathcal{O}(h^{\infty})$, its asymptotic behavior as $h\to 0$ is governed by the contributions of the critical points of the phase function $y \mapsto \psi(x,y)$ and the singular point $y=0$ of $y\mapsto y^{i\mu h}$. 

Notice first that for $x<0$, there is no positive critical points of $y\mapsto \psi(x,y)$, hence $u_1^+$ is microlocally $0$ on $\Gamma_{q_1}^{-}$.

For $x>0$, the function $y \mapsto \psi(x,y)$ has a positive non degenerate critical point $y_c(x)$ which behaves like 
$$
y_c(x) = \frac{\sqrt{D}}{\beta}x + \mathcal{O}(x^2) \;\;\; {\rm as}\;\; x\to 0.
$$
In particular, the corresponding critical value $\psi(x,y_c(x))$ coincides with the generating function $\phi_{q_1}(x)$ of $\Gamma_{q_1}$, and we have 
$$
\psi(x,y_c(x))= -\frac{\alpha}{2\beta}x^2 + \mathcal{O}(x^3) \;\;\; {\rm as}\;\; x\to 0.
$$
Moreover, we have $\partial_y^2\psi(x,y_c(x))= -\frac{\beta}{\delta}+ \mathcal{O}(x)<0$. Then, by the stationary phase theorem (see e.g. \cite{Ma} Corollary 2.6.3), the contribution of this critical point to the integral \eqref{SPO} is of the form
\begin{equation}\label{FCI}
\sqrt{2\pi h} \,\sigma^+(x;h) e^{i\phi_{q_1}(x)/h},
\end{equation}
where $\sigma^{+}(x;h)\sim \sum_{k\geq 0} h^k \sigma^+_{k}(x)$ is a symbol with leading term 
\begin{equation}
\sigma_0^+(x) = e^{-i\frac{\pi}{4}} \vert \partial_y^2\psi(x,y_c(x)) \vert^{-\frac{1}{2}} c_0(x,y_{c}(x)) = \sqrt{\frac{\delta}{\beta}}  e^{-i\frac{\pi}{4}} + \mathcal{O}(x).
\end{equation}

On the other hand, by a change of contour of integration which reduces the integral to a Laplace transform, one can see that the contribution of the endpoint $y=0$ to the integral is both for $x>0$ and $x<0$ of the form
$$
h^{1+i\mu h} \eta_+(x;h) e^{i\phi_{q_2}(x)/h},
$$
where $\eta_+(x;h) \sim \sum_{k\geq 0} h^k \eta^+_{k}(x)$ with 
$$
\eta_0^+(x) = \frac{i\delta}{\sqrt{D}x}c_0(x,0) = \frac{i\delta}{\sqrt{D}x} (1+ \mathcal{O}(x)).
$$
\end{proof}

Now, we construct another pair of solutions $v^{\pm}={}^t(v_1^{\pm},v_2^{\pm})$ to the system \eqref{GSystem} that are microlocally zero on one of $\Gamma_{q_2}^+$ and $\Gamma_{q_2}^-$. To do this, we proceed in a similar way as above but now by reducing the system \eqref{GSystem} to a scalar equation of $v_2$ instead of $v_1$.  Setting $v={}^t(v_1,v_2)$ and using the ellipticity of $\mathcal{R}^*$ at $\rho_0$, the system $\mathcal{Q}v\sim 0$ is reduced microlocally near $\rho_0$ to 
\begin{equation}\label{FR2}
\left\{\begin{array}{lll}
\widehat{\mathcal{L}} v_2 \sim 0,\\
v_1 \sim - h^{-1} (\mathcal{R}^*)^{-1} \mathcal{Q}_2 v_2,
\end{array}\right.
\end{equation}
where $(\mathcal{R}^*)^{-1} $ denotes a parametrix of $\mathcal{R}^*$ in a neighborhood of $\rho_0$ and $\widehat{\mathcal{L}}$ is the $h$-pseudodifferential operator defined by 
$$
\widehat{\mathcal{L}} := \mathcal{R}^* \mathcal{Q}_1 (\mathcal{R}^*)^{-1}  \mathcal{Q}_2 - h^2 \mathcal{R}^* (\mathcal{R}^*)^{-1} .
$$
By means of a normal form similar to \eqref{FR}, the equation $\widehat{\mathcal{L}} v_2 \sim 0$ is reduced microlocally near $\rho_0$  to $\mathcal{G} \widehat{v}_2 \sim \widehat{F}(0;h) \widehat{v}_2$ with $ \widehat{v}_2:= \widehat{U} v_2$, where the FIO $\widehat{U}$ is given by 
$$
\widehat{U}^{-1} u(x;h) = \int_{\mathbb R} e^{i \widehat{\psi}(x,y)/h} \widehat{c}(x,y;h) u(y) dy,
$$
with a symbol $\widehat{c}(x,y;h)\sim \sum_{k\geq 0}h^k \widehat{c}_k(x,y)$ satisfying $\widehat{c}_0(\rho_0)=1$, and the phase $\widehat{\psi}$ is of the form 
$$
\widehat{\psi}(x,y) = \frac{1}{2\beta}(-\alpha x^2+ 2 \sqrt{D}xy + \delta y^2) + \mathcal{O}((x,y)^3).
$$
The formal symbol $\widehat{F}(t;h)$ at $t=0$ is of the form $\widehat{F}(0;h)=-\frac{i}{2} h + \widehat{\mu} h^2$, with $\widehat{\mu}=\widehat{\mu}(h)\sim \sum_{k\geq 0} h^k \widehat{\mu}_{-,k}$ is a classical symbol of order $0$. Then, as above, setting 
$$
v_2^{\pm}:= \widehat{U}^{-1} g_{\widehat{\mu}}^{\pm},
$$
with $g_{\widehat{\mu}}^{\pm}$ defined by \eqref{gfunctions} with $\widehat{\mu}$ instead of $\mu$, we see that $v_2^{\pm}$ are microlocal solutions to the equation $\widehat{\mathcal{L}}v_2\sim 0$ microlocally in a small neighborhood of $\rho_0$ that we still denote by $\mathcal{V}$, and we have 
$$
{\rm FS}(v_2^{\pm}) \cap \mathcal{V} \subset  \big( \Gamma_{q_2}^{\pm} \cup \Gamma_{q_1} \big) \cap \mathcal{V}.
$$
Moreover, as in Proposition \ref{pre1}, we have the following asymptotic formula for $v_2^{\pm}$.
\begin{proposition}\label{pre2}
There exist symbols 
$$
\widehat{\sigma}^{\pm}(x;h)\sim \sum_{k\geq 0}h^k \widehat{\sigma}_k^{\pm}(x), \;\;\; \widehat{\eta}^{\pm}(x;h)\sim \sum_{k\geq 0}h^k \widehat{\eta}_k^{\pm}(x),
$$ 
with leading terms given by 
$$
\widehat{\sigma}_0^+(x) = \sqrt{ \frac{\beta}{\delta} } e^{i\frac{\pi}{4}} + \mathcal{O}(x) , \quad  \widehat{\eta}_0^+(x) = \frac{i \beta}{\sqrt{D}x}(1+\mathcal{O}(x)),
$$
$$
\widehat{\sigma}_0^-(x) = \sqrt{\frac{\beta}{\delta}} e^{i\frac{\pi}{4}} + \mathcal{O}(x) , \quad  \widehat{\eta}_0^-(x) = -\frac{i \beta}{\sqrt{D}x}(1+\mathcal{O}(x)),
$$
such that, modulo $\mathcal{O}(h^{\infty})$ as $h\to 0^+$, we have 

$$
v_2^+(x;h) = \left\{\begin{array}{lll}
 h^{1+ i\widehat{\mu}h} \, \widehat{\eta}_+(x;h) e^{i\phi_{q_1}(x)/h} \;\;\; & (x>0) \\ \\
\sqrt{2\pi h}\, \widehat{\sigma}^+(x;h) e^{i\phi_{q_2}(x)/h} + h^{1+ i\widehat{\mu}h} \widehat{\eta}_+(x;h) e^{i\phi_{q_1}(x)/h} \;\;\; &(x<0) \
\end{array}\right.
$$
and
$$
v_2^-(x;h) = \left\{\begin{array}{lll}
\sqrt{2\pi h} \, \widehat{\sigma}^-(x;h) e^{i\phi_{q_2}(x)/h} + h^{1+ i\widehat{\mu}h} \widehat{\eta}_-(x;h) e^{i\phi_{q_1}(x)/h} \;\;\; & (x>0) \\ \\
h^{1+ i\widehat{\mu} h} \, \widehat{\eta}_-(x;h) e^{i\phi_{q_1}(x)/h} \;\;\; & (x<0).
\end{array}\right.
$$
\end{proposition}

Summing up, we then have constructed $4$ microlocal solutions to the system \eqref{GSystem} microlocally in a small neighborhood $\mathcal{V}$ of $\rho_0$
$$
u^{\pm}={}^t(u_1^{\pm},u_2^{\pm}), \quad v^{\pm}={}^t(v_1^{\pm},v_2^{\pm}),
$$
with 
$$
{\rm FS}(u^{\pm}) \cap \mathcal{V} \subset  \big(\Gamma_{q_1}^{\pm} \cup \Gamma_{q_2}\big) \cap \mathcal{V},
$$
$$
{\rm FS}(v^{\pm}) \cap \mathcal{V} \subset  \big( \Gamma_{q_2}^{\pm} \cup \Gamma_{q_1} \big) \cap \mathcal{V},
$$
where $u_1^{\pm}$ and $v_2^{\pm}$ are defined above and 
$$
u_2^{\pm} \sim -h^{-1} \mathcal{R}^{-1} \mathcal{Q}_1 u_1^{\pm},\;\; v_1^{\pm} \sim -h^{-1} (\mathcal{R}^*)^{-1} \mathcal{Q}_2 v_2^{\pm}.
$$
\subsubsection{Proof of theorem \ref{TMF}} Now we connect our microlocal solutions $u^{\pm}$ and $v^{\pm}$ to the WKB solutions $f_{q_j}^{\pm}$, $j=1,2$, given by Proposition \ref{GenWKB} and we deduce the transfer matrix at the crossing point. The following result is an immediate consequence of Propositions \ref{pre1} and \ref{pre2}.
\begin{proposition}
There exist symbols
$$
A_{q_1}^{\pm}(h) \sim \sum_{k\geq 0} h^kA_{q_1,k}^{\pm}, \; A_{q_2}^{\pm,\pm}(h) \sim \sum_{k\geq 0} h^kA_{q_2,k}^{\pm,\pm}, \;  B_{q_1}^{\pm,\pm}(h) \sim \sum_{k\geq 0} h^kB_{q_1,k}^{\pm,\pm}, \; B_{q_2}^{\pm}(h) \sim \sum_{k\geq 0} h^kB_{q_2,k}^{\pm},
$$
with leading terms 
\begin{equation}\label{ledf1}
A_{q_1,0}^+ = A_{q_1,0}^- = \sqrt{\frac{2\pi \delta}{\beta}} e^{-i\frac{\pi}{4}}, \;\; A_{q_2,0}^{+,+} = A_{q_2,0}^{+,-}= -A_{q_2,0}^{-,+} = - A_{q_2,0}^{-,-}= \frac{i\sqrt{D}}{r(\rho_0)},
\end{equation}
\begin{equation}\label{ledf2}
B_{q_2,0}^+ = B_{q_2,0}^- = \sqrt{\frac{2\pi \beta}{\delta}} e^{i\frac{\pi}{4}}, \;\; - B_{q_1,0}^{+,+} = - B_{q_1,0}^{+,-}= B_{q_1,0}^{-,+} = B_{q_1,0}^{-,-}= \frac{i\sqrt{D}}{\overline{r(\rho_0)}},
\end{equation}
such that
$$
u^+ \sim \left\{ \begin{array}{lll}
A_{q_1}^+ h^{\frac12} f_{q_1}^+ \;\;\; & {\rm on} \;\; \Gamma_{q_1}^+ \cap \mathcal{V} \\\\
0   \;\;\; & {\rm on} \;\; \Gamma_{q_1}^- \cap \mathcal{V} \\ \\
A_{q_2}^{+,+} h^{i\mu h} f_{q_2}^{+} \;\;\; & {\rm on} \;\; \Gamma_{q_2}^{+} \cap \mathcal{V} \\\\
A_{q_2}^{+,-} h^{i\mu h} f_{q_2}^{-} \;\;\; & {\rm on} \;\; \Gamma_{q_2}^{-} \cap \mathcal{V} 
\end{array}\right. , \quad 
u^- \sim \left\{ \begin{array}{lll}
0 \;\;\; & {\rm on} \;\; \Gamma_{q_1}^+ \cap \mathcal{V} \\\\
A_{q_1}^- h^{\frac12} f_{q_1}^+  \;\;\; & {\rm on} \;\; \Gamma_{q_1}^- \cap \mathcal{V} \\ \\
A_{q_2}^{-,+} h^{i\mu h} f_{q_2}^{+} \;\;\; & {\rm on} \;\; \Gamma_{q_2}^{+} \cap \mathcal{V} \\\\
A_{q_2}^{-,-} h^{i\mu h} f_{q_2}^{-} \;\;\; & {\rm on} \;\; \Gamma_{q_2}^{-} \cap \mathcal{V}
\end{array}\right. 
$$
$$
v^+ \sim \left\{ \begin{array}{lll}
B_{q_2}^+ h^{\frac12} f_{q_2}^+ \;\;\; & {\rm on} \;\; \Gamma_{q_2}^+ \cap \mathcal{V} \\\\
0   \;\;\; & {\rm on} \;\; \Gamma_{q_2}^- \cap \mathcal{V} \\ \\
B_{q_1}^{+,+} h^{i\widehat{\mu} h} f_{q_1}^{+} \;\;\; & {\rm on} \;\; \Gamma_{q_1}^{+} \cap \mathcal{V}\\\\
B_{q_1}^{+,-} h^{i\widehat{\mu} h} f_{q_1}^{-} \;\;\; & {\rm on} \;\; \Gamma_{q_1}^{-} \cap \mathcal{V}
\end{array}\right. , \quad 
v^- \sim \left\{ \begin{array}{lll}
0 \;\;\; & {\rm on} \;\; \Gamma_{q_2}^+ \cap \mathcal{V} \\\\
B_{q_2}^- h^{\frac12} f_{q_2}^+  \;\;\; & {\rm on} \;\; \Gamma_{q_2}^- \cap \mathcal{V} \\ \\
B_{q_1}^{-,+} h^{i\widehat{\mu} h} f_{q_1}^{+} \;\;\; & {\rm on} \;\; \Gamma_{q_1}^{+} \cap \mathcal{V} \\\\
B_{q_1}^{-,-} h^{i\widehat{\mu} h} f_{q_1}^{-} \;\;\; & {\rm on} \;\; \Gamma_{q_1}^{-} \cap \mathcal{V}.
\end{array}\right. 
$$
\end{proposition}

We set 
$$
\begin{pmatrix}
t_{1}^{+}(h) \\\\
t_{2}^{+}(h)
\end{pmatrix}
= \begin{pmatrix}
s_{1,1}(h) & s_{1,2}(h) \\\\
s_{2,1}(h) & s_{2,2}(h)
\end{pmatrix}
\begin{pmatrix}
t_{1}^{-}(h) \\\\
t_{2}^{-}(h)
\end{pmatrix}.
$$
Observe that if $t_1^-(h) = 1$ and $t_2^-(h)=0$ then $u$ should be equal to $(B_{q_1}^{+,-} h^{i\widehat{\mu} h})^{-1} v^+$ microlocally near $\rho_0$, and therefore we have 
$$
s_{2,1}(h) = t_2^+(h) =  h^{\frac12 - i \widehat{\mu} h}  \kappa_{2,1}(h) \;\;\; {\rm with} \;\; \kappa_{2,1}(h) := \frac{B_{q_2}^+(h)}{B_{q_1}^{+,-}(h)},
$$
$$
s_{1,1}(h) = t_1^+(h) =  \frac{B_{q_1}^{+,+}(h)}{B_{q_1}^{+,-}(h)}.
$$

Analogously, if $t_1^-(h) = 0$ and $t_2^-(h)=1$ then $u$ should be equal to $(A_{q_2}^{+,-} h^{i\mu h})^{-1} u^+$ microlocally near $\rho_0$, and therefore we have 
$$
s_{1,2}(h) = t_1^+(h) =  h^{\frac12 - i \mu h}  \kappa_{1,2}(h) \;\;\; {\rm with} \;\; \kappa_{1,2}(h) := \frac{A_{q_1}^+(h)}{A_{q_2}^{+,-}(h)},
$$
$$
s_{2,2}(h) = t_2^+(h) =  \frac{A_{q_2}^{+,+}(h)}{A_{q_2}^{+,-}(h)}.
$$
This ends the proof of Theorem \ref{TMF}.

\section{Microlocal connection formulae}

In this section, we establish the microlocal connection formulae for our operator \eqref{System0}. In our case, the characteristic set ${\rm Char}\,(P(h)-E)=\Gamma_1(E)\cup \Gamma_2(E)$ is divided  by $2$ \textit{crossing points} $\rho_{\pm}(E)=(0,\pm \sqrt{E})$ and $4$ \textit{turning points (caustics)} $(\gamma(E),0)$, $\gamma=\alpha_1,\alpha_2,\beta_1,\beta_2$,
into the following $8$ curves $\Gamma_{j,S}^{\pm}$ (see Figure \ref{Fig 3})
$$
\Gamma_{j,L}:=\{(x,\xi)\in \Gamma_j(E); x<0\},\quad
\Gamma_{j,L}^\pm:=\{(x,\xi)\in \Gamma_j(E); x<0,\pm\xi>0\},
$$
$$
\Gamma_{j,R}:=\{(x,\xi)\in \Gamma_j(E); x>0\},\quad
\Gamma_{j,R}^\pm:=\{(x,\xi)\in \Gamma_j(E); x>0,\pm\xi>0\}.
$$

\begin{center}
\begin{figure}[h]
	\begin{tikzpicture}[scale=1.2]
	\begin{scope}[decoration={markings,
	mark=at position 0.05 with {\arrowreversed[very thick]{>}},
    mark=at position 0.25 with {\arrowreversed[very thick]{>}},
     mark=at position 0.75 with {\arrowreversed[very thick]{>}},
                   mark=at position 0.95 with {\arrowreversed[very thick]{>}},
    },
            ]
	\draw[very thick,postaction={decorate}] \boundellipse{-0.9,0}{2.5}{1.5} ;
	\end{scope}
	\begin{scope}[decoration={markings,
    mark=at position 0.44 with {\arrowreversed[very thick]{>}},
        mark=at position 0.20 with {\arrowreversed[very thick]{>}},
    mark=at position 0.55 with {\arrowreversed[very thick]{>}},
        mark=at position 0.75 with {\arrowreversed[very thick]{>}}},
              ]
	\draw[very thick,postaction={decorate}] \boundellipse{1.89,0}{2.5}{1.5};
	\end{scope}
	\node at (3,1.7) {\tiny{$\Gamma_{2,R}^+$}};
	\node at (-2,1.7) {\tiny{$\Gamma_{1,L}^+$}};
	\node at (-2,-1.7) {\tiny{$\Gamma_{1,L}^-$}};
	\node at (1.7,0.8) {\tiny{$\Gamma_{1,R}^+$}};
	\node at (1.7,-0.8) {\tiny{$\Gamma_{1,R}^-$}};
	\node at (3,-1.7) {\tiny{$\Gamma_{2,R}^-$}};
	\node at (-0.8,0.8) {\tiny{$\Gamma_{2,L}^+$}};
	\node at (-0.8,-0.8) {\tiny{$\Gamma_{2,L}^-$}};
	\node at (0.56,1.6) {\tiny{$\rho_+(E)$}};
	\node at (0.56,-1.6) {\tiny{$\rho_-(E)$}};
         \node at (-3.99,0) {\tiny{$(\alpha_1(E),0)$}};
          \node at (-1.2,0) {\tiny{$(\alpha_2(E),0)$}};
         \node at (2.2,0) {\tiny{$(\beta_1(E),0)$}};
         \node at (4.98,0) {\tiny{$(\beta_2(E),0)$}};
                   \draw [fill=white]  (0.51,1.24) circle[radius= 0.15 em];
                   \draw [fill=white] (0.51,-1.24) circle[radius= 0.15 em];
                   \draw [fill=white] (-3.4,0) circle[radius= 0.15 em];
                   \draw [fill=white] (4.4,0) circle[radius= 0.15 em];
                   \draw [fill=white] (1.6,0) circle[radius= 0.15 em];
          \draw [fill=white]  (-0.6,0) circle[radius= 0.15 em];
	\end{tikzpicture}
	\caption{The $8$ curves $\Gamma_{j,S}^{\pm}$, $j=1,2$, $S=L,R$.} \label{Fig 3}
	\end{figure}
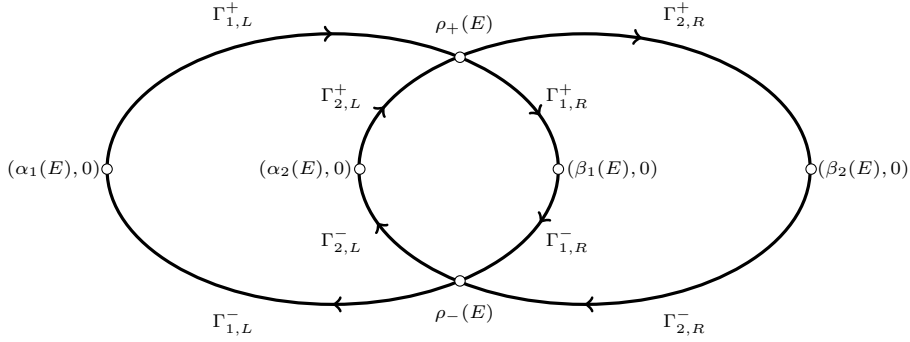
\end{center}

In a similar way as in Proposition \ref{GenWKB}, we have the following WKB basis of microlocal solutions to the system $(P(h)-E)u=0$ on each curve $\Gamma_{j,S}^{\pm}$, $j=1,2$, $S=L,R$.
\begin{proposition}\label{wkb-}
Let $E\in I_{0}$. On each of the $8$ curves $\Gamma_{j,S}^{\pm}$, $j=1,2$, $S=L,R$, the space of microlocal solutions to the system 
$$
(P(h)-E)u\sim 0 \,\,\,{\rm microlocally\, on}\,\,
\Gamma_{j,S}^{\pm},
$$
is one-dimensional, and there exist $8$ functions $f_{1,L}^{\pm}, f_{1,R}^{\pm}, f_{2,L}^{\pm}, f_{2,R}^{\pm}\in L^2(\R)$ such that 
$$
(P(h)-E)f_{1,L}^{\pm}\sim 0,\quad f_{1,L}^{\pm}\sim 
\begin{pmatrix}
a_1^{\pm} \\\\
ha_2^{\pm}
\end{pmatrix}
e^{\pm i\phi_1(x)/h}
\,\,\,{\rm microlocally\, on}\,\,
\Gamma_{1,L}^{\pm}(E),
$$
$$
(P(h)-E)f_{1,R}^{\pm}\sim 0,\quad f_{1,R}^{\pm}\sim 
\begin{pmatrix}
a_1^{\pm} \\\\
ha_2^{\pm}
\end{pmatrix}e^{\pm i\phi_1(x)/h}
\,\,\,{\rm microlocally\, on}\,\,
\Gamma_{1,R}^{\pm}(E),
$$
$$
(P(h)-E)f_{2,L}^{\pm}\sim 0,\quad f_{2,L}^{\pm}\sim 
\begin{pmatrix}
hb_1^{\pm} \\\\
b_2^{\pm}
\end{pmatrix}
e^{\pm i\phi_2(x)/h}
\,\,\,{\rm microlocally\, on}\,\,
\Gamma_{2,L}^{\pm}(E),
$$
$$
(P(h)-E)f_{2,R}^{\pm}\sim 0,\quad f_{2,R}^{\pm}\sim 
\begin{pmatrix}
hb_1^{\pm} \\\\
b_2^{\pm}
\end{pmatrix}
e^{\pm i\phi_2(x)/h}
\,\,\,{\rm microlocally\, on}\,\,
\Gamma_{2,R}^{\pm}(E),
$$
where for $j=1,2$,
\begin{equation}\label{phase function}
\phi_j(x):= \int_0^x \sqrt{ E-V_j(t)} \,dt, 
\end{equation}
and $a_j^{\pm}= a_j^{\pm}(x;h)\sim \sum_{k\geq 0}h^ka_{j,k}^{\pm}(x)$ and $b_j^{\pm}=b_j^{\pm}(x;h)\sim \sum_{k\geq 0}h^kb_{j,k}^{\pm}(x)$ are symbols whose first coefficients are given by

\begin{equation*}\label{coeffaj}
\begin{aligned}
& a_{1,0}^{\pm}(x)=\frac1{(E-V_1(x))^{\frac14}} \quad ; \quad  b_{2,0}^{\pm}(x)=\frac1{(E-V_2(x))^{\frac14}} ,
\end{aligned}
\end{equation*}
\begin{equation*}\label{coeffbj}
\begin{aligned}
& a_{2,0}^{\pm}(x)=\frac{r_0(x)\mp ir_1(x)\sqrt{E-V_1(x)}}{(V_1(x)-V_2(x))(E-V_1(x))^{\frac14}} \quad ; \quad  b_{1,0}^{\pm}(x) =\frac{r_0(x)\pm ir_1(x)\sqrt{E-V_2(x)}}{(V_2(x)-V_1(x))(E-V_2(x))^{\frac14}}.
\end{aligned}
\end{equation*}
\end{proposition}

Now, to know the global behavior of solutions to the system $(P(h)-E)u=0$, we compute the connection formulae between these WKB solutions at the crossing and turning points.

\subsection{Transfer Matrices at crossing points} By applying the results of the previous section, more precisely Theorem \ref{TMF}, to our system $P(h)-E$ with $q_j(x,\xi)=p_j(x,\xi)-E=\xi^2+ V_j(x)-E$, $j=1,2$, and $r(x,\xi)=r_0(x)+ i r_1(x)\xi$, we obtain the following transfert matrices at the crossing points $\rho_{\pm}(E)$.

\begin{proposition}[\textbf{Transfer matrix at $\rho_-(E)$}]\label{connection1}
Let $u(x;h)\in L^2(\R)$ be a solution to $(P(h)-E)u\sim 0$ microlocally in a small neighbourhood of $\rho_-(E)$ and set
\begin{equation}
u \sim \left\{ \begin{array}{lll}  t_{1,L}^{-}f_{1,L}^{-}
\,\,\,{\rm microlocally\, on}\,\,
\Gamma_{1,L}^{-}(E) \\
 t_{1,R}^{-}f_{1,R}^{-}
\,\,\,{\rm microlocally\, on}\,\,
\Gamma_{1,R}^{-}(E) \\
t_{2,L}^{-}f_{2,L}^{-}
\,\,\,{\rm microlocally\, on}\,\,
\Gamma_{2,L}^{-}(E) \\
t_{2,R}^{-}f_{2,R}^{-}
\,\,\,{\rm microlocally\, on}\,\,
\Gamma_{2,R}^{-}(E)\\
\end{array}\right.
\end{equation}
for some constants $t_{j,S}^-=t_{j,S}^-(E;h)$. Then it holds that
\begin{equation}
\label{tau-}
\begin{pmatrix}
t_{1,R}^{-} \\\\
t_{2,L}^{-}
\end{pmatrix}
=
\mathcal{M}_-
\begin{pmatrix}
t_{1,L}^{-} \\\\
t_{2,R}^{-}
\end{pmatrix},
\end{equation}
where $\mathcal{M}_-= \mathcal{M}_-(E;h)$ is the $2\times 2$ matrix given by 
$$
\mathcal{M}_-(E;h):= \begin{pmatrix}
\kappa_{1,1}^-(E;h) & \kappa_{1,2}^-(E;h) \\\\
\kappa_{2,1}^-(E;h) & \kappa_{2,2}^-(E;h)
\end{pmatrix}
$$ 
and the coefficients $\kappa_{j,k}^-$ have the following asymptotic behaviors 
$$
\kappa_{1,1}^-(E;h)= 1+ \mathcal{O}(h) , \quad \kappa_{2,2}^-(E;h)=1+ \mathcal{O}(h),
$$
$$
\kappa_{1,2}^-(E;h) = \tau_0 h^{\frac12} + \mathcal{O}(h^{\frac32} \ln(1/h)),
$$
$$
\kappa_{2,1}^-(E;h) =  \overline{\tau_0} h^{\frac12} + \mathcal{O}(h^{\frac32} \ln(1/h)),
$$
uniformly for $E\in I_{0}$ and $h\to 0^+$, with 
\begin{equation}
\tau_0:= e^{i\frac{\pi}{4}} \sqrt{ \frac{\pi}{V_1'(0)-V_2'(0)}} (r_0(0)E^{-\frac14} - i r_1(0)E^{\frac14}).
\end{equation}

\end{proposition}

\begin{proposition}[\textbf{Transfer matrix at $\rho_+(E)$}]
\label{connection2}
Let $u(x;h)\in L^2(\R)$ be a solution to $(P(h)-E)u\sim 0$ microlocally in a small neighbourhood ${\mathcal V}_+$ of $\rho_+(E)$ and set
\begin{equation}
u \sim \left\{ \begin{array}{lll}
 t_{1,L}^{+}f_{1,L}^{+}
\,\,\,{\rm microlocally\, on}\,\,
\Gamma_{1,L}^{+}(E)\\
t_{1,R}^{+}f_{1,R}^{+}
\,\,\,{\rm microlocally\, on}\,\,
\Gamma_{1,R}^{+}(E)\\
t_{2,L}^{+}f_{2,L}^{+}
\,\,\,{\rm microlocally\, on}\,\,
\Gamma_{2,L}^{+}(E)\\
t_{2,R}^{+}f_{2,R}^{+}
\,\,\,{\rm microlocally\, on}\,\,
\Gamma_{2,R}^{+}(E),
\end{array}\right.
\end{equation}
for some constants $t_{j,S}^+=t_{jS}^+(E;h)$. Then  it holds that
\begin{equation}\label{tau+}
\begin{pmatrix}
t_{1,L}^{+} \\ \\
t_{2,R}^{+} 
\end{pmatrix} = \mathcal{M}_+\begin{pmatrix}
t_{1,R}^{+}\\\\
t_{2,L}^{+} 
\end{pmatrix}
\end{equation}
where $\mathcal{M}_+= \mathcal{M}_+(E;h)$ is the $2\times 2$ matrix given by 
$$
\mathcal{M}_+(E;h):=\begin{pmatrix}
\kappa_{1,1}^+(E;h) & \kappa_{1,2}^+(E;h) \\ \\
\kappa_{2,1}^+(E;h) & \kappa_{2,2}^+(E;h)
\end{pmatrix}
$$
and the coefficients $\kappa_{j,k}^+(E;h)$ have the following asymptotic behaviors 
$$
\kappa^+_{1,1}(E;h)=1+ \mathcal{O}(h),\quad \kappa^+_{2,2}(E;h)= 1+ \mathcal{O}(h) ,
$$
$$
\kappa^+_{1,2}(E;h)= - \overline{\tau_0} h^{\frac12} + \mathcal{O}(h^{\frac32} \ln(1/h)),
$$
$$
\kappa^+_{2,1}(E;h)=-\tau_0 h^{\frac12} + \mathcal{O}(h^{\frac32} \ln(1/h)),
$$
uniformly for $E\in I_{0}$ and $h\to 0^+$.
\end{proposition}

\subsection{Microlocal connection formulae at turning points}
It remains now to compute the connection formulae at the turning points. We define the action integrals 
$$
S_{j,L}(E):=\int_{\alpha_j(E)}^0\sqrt{E-V_j(x)}dx,\quad S_{j,R}(E):=\int_0^{\beta_j(E)}\sqrt{E-V_j(x)}dx \quad (j=1,2).
$$
\begin{proposition}\label{maslov}
Let $u$ be a solution to $(P(h)-E)u\sim 0$ microlocally near $\Gamma_{j,S}$, $j=1,2$, $S=L,R$, and suppose that
\begin{equation}
u \sim \left\{ \begin{array}{lll} 
t_{j,S}^{+}f_{j,S}^{+}
\,\,\,{\rm microlocally\, on}\,\,
\Gamma_{j,S}^{+}(E)\\
t_{j,S}^{-}f_{j,S}^{-}
\,\,\,{\rm microlocally\, on}\,\,
\Gamma_{j,S}^{-}(E),
\end{array}\right.
\end{equation}
for some constants $t_{j,S}^{\pm}=t_{j,S}^{\pm}(E;h)$. Then it holds that
\begin{equation}
\label{sigmaS}
t_{j,R}^{+}=\mathcal{T}_{j,R}\; t_{j,R}^{-}, \quad t_{j,L}^{-}=\mathcal{T}_{j,L}\; t_{j,L}^{+},
\end{equation}
with constants $\mathcal{T}_{j,S}=\mathcal{T}_{j,S}(E;h)$, which behave, as $h\to 0^+$,
$$
\mathcal{T}_{j,S}=ie^{ { -2iS_{j,S}/h}}+\ord (h), \;\;\; j=1,2,S=L,R.
$$
\end{proposition}

 The proof of this result relies on the fact that microlocally near a turning point $(\gamma_j(E),0)$, $j=1,2$, $\gamma=\alpha,\beta$, the operator $P_k(h)-E$, $j\neq k\in \{1,2\}$, is elliptic. The system \eqref{System} is then reduced to a scalar equation and the idea is to solve it using Maslov's method. We refer to \cite[Lemma 6.1]{FMW3} for a detailed proof. 

\section{Proofs of Theorems \ref{thQCS} and \ref{mainth}} 


\subsection{Proof of Theorem \ref{thQCS}} Let $E=E(h)\in I_{0}$ be eigenvalue of $P(h)$  and $u=u(x;h)\in L^2(\mathbb R)$ a corresponding eigenfunction. Suppose that  
\begin{equation}
u \sim \left\{ \begin{array}{lll} 
t_{1,L}^-(E;h) f_{1,L}^- \quad \text{microlocally on} \;\; \Gamma_{1,L}^-(E) \\
t_{2,R}^-(E;h) f_{2,R}^- \quad \text{microlocally on} \;\; \Gamma_{2,R}^-(E),
\end{array}\right.
\end{equation}
for some complex numbers $t_{1,L}^-(E;h), t_{2,R}^-(E;h)$. Then, it follows by the connection formulae of Propositions \ref{connection1}, \ref{connection2} and \ref{maslov}, that
\begin{equation}
u \sim \left\{ \begin{array}{lll} 
t_{1,L}^+(E;h) f_{1,L}^+ \quad \text{microlocally on} \;\; \Gamma_{1,L}^+(E) \\
t_{2,R}^+(E;h) f_{2,R}^+ \quad \text{microlocally on} \;\; \Gamma_{2,R}^+(E),
\end{array}\right.
\end{equation}
with 
\begin{equation}\label{expr1}
\begin{pmatrix} t_{1,L}^+(E;h)\\ \\ t_{2,R}^+(E;h)  \end{pmatrix} = \mathcal{M}_+(E;h) \begin{pmatrix} \mathcal{T}_{1,R}(E;h)  &   0 \\\\
0 &  (\mathcal{T}_{2,L}(E;h) )^{-1} \end{pmatrix} \mathcal{M}_-(E;h) \begin{pmatrix} t_{1,L}^-(E;h)\\\\ t_{2,R}^-(E;h)  \end{pmatrix} .
\end{equation}
On the other hand, by passing through the turning points $(\alpha_1(E),0)$ and $(\beta_2(E),0)$ using the relations \eqref{sigmaS}, we obtain 
\begin{equation}\label{expr2}
\begin{pmatrix} t_{1,L}^+(E;h)\\ \\ t_{2,R}^+(E;h)  \end{pmatrix}  = \begin{pmatrix} (\mathcal{T}_{1,L}(E;h))^{-1} & 0 \\ \\ 0 & \mathcal{T}_{2,R}(E;h)\end{pmatrix}\begin{pmatrix} t_{1,L}^-(E;h)\\ \\  t_{2,R}^-(E;h)  \end{pmatrix}.
\end{equation}
Putting together \eqref{expr1} and \eqref{expr2}, we deduce the following quantization condition for the eigenvalues of $P(h)$ in $I_0$.
\begin{proposition} $E=E(h)\in I_{0}$ is an eigenvalue of $P(h)$ if and only if $\det(\Lambda(E;h)-I_2)=0$, where $\Lambda(E;h)$ is the $2\times 2$ matrix defined by  
\begin{equation*}
\Lambda(E;h) := \begin{pmatrix} \mathcal{T}_{1,L}(E;h) & 0 \\ \\ 0 & (\mathcal{T}_{2,R}(E;h))^{-1}\end{pmatrix} \mathcal{M}_+(E;h) \begin{pmatrix} \mathcal{T}_{1,R}(E;h)  &   0 \\\\
0 &  (\mathcal{T}_{2,L}(E;h))^{-1} \end{pmatrix} \mathcal{M}_-(E;h).
\end{equation*}
\end{proposition}
Setting $\Lambda(E;h)=(\lambda_{j,k}(E;h))_{1\leq j,k\leq 2}$ and using the asymptotic formulae for the $\kappa^{\pm}_{j,k}(E;h)$'s and the $\mathcal{T}_{j,S}(E;h)$'s given in the previous section, we obtain, 
$$
\lambda_{1,1}(E;h) = - e^{-2i \mathcal{A}_1(E)/h}  + m_{1,1}(E;h),
$$
$$
\lambda_{2,2}(E;h) = - e^{2i \mathcal{A}_2(E)/h} + m_{2,2}(E;h),
$$
$$
\lambda_{1,2}(E;h) =  \lambda_{1,2}^0(E) h^{\frac12}  + m_{1,2}(E;h),
$$
$$
\lambda_{2,1}(E;h) = \lambda_{2,1}^0(E) h^{\frac12}  + m_{2,1}(E;h),
$$
with 
$$
\lambda_{1,2}^0(E) := - \left( \tau_0 e^{-2i \mathcal{A}_1(E)/h} +  \overline{\tau_0} e^{2i(S_{2,L}(E)-S_{1,L}(E))/h}    \right),
$$
$$
\lambda_{2,1}^0(E) := - \left( \overline{\tau_0} e^{2i \mathcal{A}_2(E)/h} + \tau_0 e^{2i(S_{2,R}(E)-S_{1,R}(E))/h}    \right),
$$
and $m_{j,k}(E;h)$ are analytic symbols in $I_0$ satisfying the estimates
$$
m_{j,k}(E;h) = \left\{\begin{array}{lll}
\mathcal{O}(h) \;\;\; & j=k  \\
\mathcal{O}(h^{\frac32}\ln(1/h)) \;\;\; & j\neq k,
\end{array}\right.
$$
uniformly for $E\in I_{0}$ and $h\to 0^+$. This immediately implies \eqref{precise quantification condition}, with $m_1(E;h) = - \frac{1}{2}e^{i\mathcal{A}_1(E)/h} m_{1,1}(E;h)$, $m_2(E;h) = - \frac{1}{2}e^{-i\mathcal{A}_2(E)/h} m_{2,2}(E;h)$ and 
$$
m_0(E;h):=  \frac{1}{4} e^{i(\mathcal{A}_1(E)-\mathcal{A}_2(E))/h} \lambda_{1,2}(E;h) \lambda_{2,1}(E;h).
$$ 
In particular, we have 
$$
m_0(E;h) = e^{i(\mathcal{A}_1(E)-\mathcal{A}_2(E))/h}  \mathcal{D}(E)  h + \mathcal{O}(h^{2}\ln(1/h)),
$$
uniformly for $E\in I_{0}$ and $h\to 0^+$, with $\mathcal{D}(E):= \frac{1}{4}\lambda_{1,2}^0(E) \lambda_{2,1}^0(E)$.

Now, in the symmetric case $V_1(x)=V_2(-x)$, using that $\mathcal{A}_1(E)= \mathcal{A}_2(E)=:\mathcal{A}(E)$, $S_{1,L}(E)=S_{2,R}(E)$ and $S_{2,L}(E)=S_{1,R}(E)$, we get
\begin{eqnarray*}
\mathcal{D}(E) &=& \frac{1}{4} \left\vert  \overline{\tau_0} e^{2i \mathcal{A}(E)/h} + \tau_0 e^{2i(S_{1,L}(E)-S_{1,R}(E))/h}  \right\vert^2 \\
&=& \frac{1}{4} \left\vert  \overline{\tau_0} e^{i \mathcal{B}(E)/h} + \tau_0 e^{-i \mathcal{B}(E)/h}  \right\vert^2,
\end{eqnarray*}
with $\mathcal{B}(E):= \mathcal{A}(E)- \big( S_{1,L}(E)-S_{1,R}(E) \big)$. It is then an elementary computation to get \eqref{FuncD}.

\subsection{Proof of Theorem \ref{mainth}} The quantization condition \eqref{precise quantification condition} can be rewritten as 
\begin{equation}\label{eqpro}
\cos\left( \frac{\mathcal{A}_1(E)}{h}\right) \cos\left( \frac{\mathcal{A}_2(E)}{h}\right) = \widehat{m}(E;h),
\end{equation}
where $\widehat{m}(E;h)$ is an analytic symbol with respect to $E\in I_{0}$ given by 
$$
\widehat{m}(E;h) := m_0(E;h) - m_2(E;h) \cos\left( \frac{\mathcal{A}_1(E)}{h}\right)  - m_1(E;h) \cos\left( \frac{\mathcal{A}_2(E)}{h}\right)  - (m_1m_2)(E;h).
$$
According to \eqref{smestimates}, $\widehat{m}$ satisfies the estimate $\widehat{m}(E;h)= \mathcal{O}(h)$, uniformly for $E\in I_{0}$ and $h>0$ small enough. To solve \eqref{eqpro}, we recall that the roots in $I_h$ to the equation 
$$
\cos\left( \frac{\mathcal{A}_1(E)}{h}\right) \cos\left( \frac{\mathcal{A}_2(E)}{h}\right) =0
$$ 
are given by the set of approximate eigenvalues $\mathcal{U}_h$ defined by \eqref{SetU}. Thus,  we can proceed exactly in the same way as in the proof of Theorem \ref{ThapproxRes} using Rouch\'e Theorem, by replacing $h^{\frac76}$ by $h^{\frac32}$. This leads to the existence of a bijection $b_h: \sigma\,(P(h))\cap I_{h} \rightarrow \mathcal{U}_{h}$ satisfying $b_h(E)-E=\mathcal{O}(h^\frac32)$ for $h\to 0^+$ in some subset $J\subset(0,1]$ with $0\in \overline{J}$.

It remains now to estimate the splitting of eigenvalues in the symmetric case $V_1(x)=V_2(-x)$. Set $\mathcal{A}(E):=\mathcal{A}_1(E)=\mathcal{A}_2(E)$. Then, the quantization condition \eqref{eqpro} becomes in this case 
\begin{equation}\label{qusy}
\cos^2\left( \frac{\mathcal{A}(E)}{h}\right) =  \widehat{m}(E;h),
\end{equation}
where using \eqref{coefm0}, \eqref{smestimates} and the fact that $\cos\left( \frac{\mathcal{A}(E)}{h}\right) = \mathcal{O}(h^\frac12)$, we have
$$
\widehat{m}(E;h) = \mathcal{D}(E) h + \mathcal{O}(h^\frac32),
$$
uniformly for $E\in I_{h}$ and $h>0$ small enough. By Taylor's formula, we immediately see that the solutions to the equation \eqref{qusy} in $I_{h}$ are of the form 
$$
E_{\pm}(h) = e(h) \pm \frac{\sqrt{D(e(h))}}{\mathcal{A}'(e(h))} h^\frac32 + \mathcal{O}(h^{\frac74}), \;\; e(h)\in \mathcal{U}_h,
$$
as $h\to 0^+$. This ends the proof of Theorem \ref{mainth}.

\appendix

\section{Microlocal WKB solutions}\label{MWSA}

In this appendix, we prove Proposition \ref{GenWKB} where a basis of microlocal solutions to the system \eqref{GSystem} on each curve $\Gamma_{q_j}^{\pm}$, $j=1,2$, is given. Recall that the space of microlocal solutions on each of these curves is one-dimensional (see the paragraph before the statement of the Proposition). Our construction is based on formal computations using standard pseudodifferential calculus. 
 
Le us start by computing $e^{-i \phi(x)/h} M(e^{i\phi/h} a)$, where $M$ is a pseudodifferential operator with Weyl symbol $m(x,\xi)$. We have
$$
e^{-i \phi(x)/h} M(e^{i\phi/h} a)(x) = \frac{1}{2\pi h} \int_{\mathbb R^2} e^{i(x-y)\xi/h} e^{-i(\phi(x)-\phi(y))/h} m\left(\frac{x+y}{2},\xi \right) a(y)dy d\xi.
$$
Writing $\phi(x)-\phi(y)=(x-y)\psi(x,y)$ with $\psi(x,y)=\int_0^1 \phi'(sx + (1-s)y) ds$, we have, after a change of variable $(y,\xi -\psi(x,y)) \mapsto (y,\xi)$,
\begin{align*}
e^{-i \phi(x)/h} M(e^{i\phi/h} a)(x) &= \frac{1}{2\pi h} \int_{\mathbb R^2} e^{i(x-y)\xi/h}  m\left(\frac{x+y}{2},\xi + \psi(x,y)\right) a(y)dy d\xi \\
&= \frac{1}{2\pi h} \int_{\mathbb R^2} e^{-iz\xi/h}  m\left(x+\frac{z}{2},\xi + \psi(x,x+z)\right) a(x+z)dz d\xi.
\end{align*}
The RHS is an oscillatory integral with quadratic phase $(z,\xi)\mapsto -z\xi$. Hence the stationary phase theorem (see e.g. \cite{Ma} Corollary 2.6.3) gives the asymptotic expansion 
$$
e^{-i \phi(x)/h} M(e^{i\phi/h} a)(x) \sim \sum_{k=0}^{+\infty} \frac{h^k}{i^k k!} (\partial_z \partial_{\xi})^k c(x,0,0),
$$
where $c(x,z,\xi):= m\left(x+\frac{z}{2},\xi + \psi(x,x+z)\right) a(x+z)$. In particular, using that $\psi(x,x)=\phi'(x)$ and $\partial_y\psi(x,x)=\frac12 \phi''(x)$, we obtain
\begin{equation}\label{FEQWKB}
e^{-i\phi(x)/h} M(e^{i\phi/h}a)(x) = m(x,\phi'(x)) a - ih S_{m}(x,\partial_x)a+ \mathcal{O}(h^2),
\end{equation}
where $S_{m}(x,\partial_x)$ is a first order differential operator given by 
\begin{equation}\label{opS}
S_{m}(x,\partial_x) = \partial_{\xi}m(x,\phi'(x))\partial_x + \frac12 \left( \partial_x\partial_{\xi}m(x,\phi'(x)) + \phi''(x) \partial_{\xi}^2m(x,\phi'(x)) \right).
\end{equation}
Now, we substitute \eqref{WKB form} into \eqref{GSystem}. Then, for $a_{q_1}\sim \sum h^k a_{q_1,k} $, $b_{q_1}\sim \sum h^k b_{q_1,k}$, using the above computation, we obtain,
$$
e^{-i\phi_{q_1}/h} \mathcal{Q} f_{q_1}^{\pm} \sim \begin{pmatrix}
q_1(x,\phi_{q_1}'(x)) a_{q_1} - ih S_{q_1} a_{q_1} + h r(x,\phi_{q_1}'(x)) b_{q_1} + \mathcal{O}(h^2)\\\\
 h \overline{r(x,\phi_{q_1}'(x))} a_{q_1} + q_2(x,\phi_{q_1}'(x)) b_{q_1} - ih S_{q_2} b_{q_1}+ \mathcal{O}(h^2)
\end{pmatrix}.
$$
The RHS is a power series of $h$, and in order that $f_{q_1}^{\pm}$ is a microlocal solution, each coefficient should vanish.

The coefficient of the $0$-th order is ${}^t\big( q_1(x,\phi_{q_1}'(x))a_{q_1,0}, q_2(x,\phi_{q_1}'(x))b_{q_1,0}\big)$. We are looking for a microlocal WKB solution supported on $\Gamma_{q_1}$, and the phase function $\phi_{q_1}$ should satisfy the eikonal equation \eqref{eikonalphase}. In particular, one has
$$
\phi_{q_1}'(x) = - \frac{\partial_x q_1(0,0)}{\partial_{\xi} q_1(0,0)} x + \mathcal{O}(x^2) \;\; {\rm as}\; x\to 0.
$$
Then the first entry of the above $0$-th order coefficient is $0$. In order that the second entry is $0$, $b_{q_1,0}$ should vanish since 
$$
q_2(x,\phi_{q_1}'(x)) = \frac{ \{q_1,q_2\}(0,0) }{ \partial_{\xi}q_1(0,0) } x + \mathcal{O}(x^2)
$$
does not vanish near $0$ except at $0$.

The coefficient of the $1$st order is ${}^t\big( -i S_{q_1}a_{q_1,0}, \overline{r(x,\phi_{q_1}'(x))} a_{q_1,0} + q_2(x,\phi_{q_1}'(x)) b_{q_1,1} \big)$. Hence, $a_{q_1,0}$ and $b_{q_1,1}$ should satisfy the equations
\begin{equation}\label{Sa}
S_{q_1} a_{q_1,0} = 0,
\end{equation}
\begin{equation}
b_{q_1,1} = - \frac{ \overline{r(x,\phi_{q_1}'(x))} }{q_2(x, \phi_{q_1}'(x))} f_{q_1,0}.
\end{equation}
Notice that in \eqref{opS}, the coefficient $\partial_{\xi}q_1(x,\phi'(x))= \partial_{\xi}q_1(0,0)+ \mathcal{O}(x)$ does not vanish near $x=0$. Hence $a_{q_1,0}$ is uniquely determined by \eqref{Sa} under the initial condition $a_{q_1,0}(0)=1$ and it is a non-zero analytic function near zero. More precisely, we have 
 $$
 a_{q_1,0}(x) =  \exp\left( -\int_0^x \frac{ \partial_x\partial_{\xi}q_1(t,\phi_{q_1}'(t)) + \phi_{q_1}''(t)\partial_{\xi}^2 q_1(t,\phi_{q_1}'(t))     }{ 2\partial_{\xi} q_1(t,\phi_{q_1}'(t))       } dt\right) = 1+ \mathcal{O}(x)
 $$
as $x\to 0$.

We also see from the condition \eqref{Cond model} that $b_{q_1,1}(x)$ has a pole of order one at $x=0$. In particular, 
$$
b_{q_1,1}(x) = - \frac{\partial_{\xi}q_1(0,0)\overline{r(0,0)}}{\{q_1,q_2\}(0,0)} \frac{1+ \mathcal{O}(x)}{x} \;\; {\rm as}\; x\to 0.
$$
In the same way, $a_{q_1,k}$, $k\geq 1$, is uniquely determined under $a_{q_1,k}(0)=0$ by an inhomogenious differential equation of the form $S_{q_1} a_{q_1,k}=F_k$ with $F_k$ depending on $a_{q_1,j}$, $j\leq k-1$, and $b_{q_1,j}$, $j\leq k$, and $b_{q_1,k}$ is determined algebraically from $a_{q_1,j}$, $j\leq k-1$ and $b_{q_1,j}$, $j\leq k-1$ and their derivatives. 


\section*{Acknowledgements}

The research of M. Assal was supported by CONICYT FONDECYT Grant No. 3180390, and the research of S. Fujii\'e was supported by the JSPS grant-in-aid No. 18K03384. S. Fujii\'e is grateful to the Faculty of Mathematics of the Pontificia Universidad Cat\'olica de Chile, where a part of this work was done, for his warm hospitality in December 2018. The authors thank Andr\'e Martinez for valuable discussions. The authors also thank Yves Colin de Verdi\`ere for his suggestions.

\end{document}